\documentclass[longbibliography]{amsart}
\usepackage[foot]{amsaddr}
\usepackage[unicode=true,pdfusetitle, bookmarks=true,bookmarksnumbered=false,bookmarksopen=false, breaklinks=false,pdfborder={0 0 0},backref=false,colorlinks=false]{hyperref}
\hypersetup{colorlinks,linkcolor=myurlcolor,citecolor=myurlcolor,urlcolor=myurlcolor}
\usepackage{graphics,epstopdf,graphicx,amsthm,amsmath,amssymb,mathptmx,braket,colortbl,color,bm,framed,mathrsfs}
\usepackage[up]{subfigure}
\usepackage{tikz}
\usepackage{multirow}
\usepackage{enumerate}
\usepackage{float}
\usepackage{cite}
\usepackage{fontawesome}
\usepackage{cleveref}
\usepackage{orcidlink}

\usepackage{arydshln}

\usepackage{tabularx}
\usepackage{stmaryrd}

\usepackage{booktabs} 
\let\hbar\undefined
\usepackage{fouriernc} 

\usepackage{algorithm}
\usepackage{algpseudocode}

\definecolor{myurlcolor}{rgb}{0,0,0.9}

\newcommand{\proj}[1]{| #1\rangle\!\langle #1 |}

\newcommand{\inner}[2]{\langle #1 , #2\rangle}

\DeclareMathOperator{\trace}{Tr}

\newcommand{\Ptr}[2]{\trace_{#1}\Pa{#2}}
\newcommand{\Tr}[1]{\Ptr{}{#1}}

\newcommand{\Pa}[1]{\left[#1\right]}

\theoremstyle{plain}
\newtheorem{thm}{Theorem}
\newtheorem{lem}[thm]{Lemma}
\newtheorem{prop}[thm]{Proposition}
\newtheorem{cor}[thm]{Corollary}

\newtheorem{Def}[thm]{Definition}
\newtheorem{Rem}[thm]{Remark}

\newtheorem{Examp}[thm]{Example}

\newcommand*{\myproofname}{Proof}

\def\ot{\otimes}
\def\complex{\mathbb{C}}
\def\CN{\mathbb{C}}
\def\real{\mathbb{R}}
\def\R{\mathbb{R}}
\def\Z{\mathbb{Z}}

\def\sgn{\mathrm{sgn}}

\newcommand{\CDD}{\mathcal D}

\newcommand{\BFF}{\mathbb F}

\newcommand{\CGG}{\mathcal G}

\newcommand{\CMM}{\mathcal M}

\newcommand{\CNN}{\mathcal N}

\newcommand{\CPP}{\mathcal P}

\newcommand{\bmu}{{\boldsymbol{\mu}}}

\newcommand{\bkappa}{\boldsymbol{\kappa}}
\newcommand{\be}{\begin{equation}}
\newcommand{\ee}{\end{equation}}

\renewcommand{\ge}{\geqslant}
\renewcommand{\geq}{\geqslant}
\renewcommand{\leq}{\leqslant}
\renewcommand{\le}{\leqslant}

\DeclareMathAlphabet{\mathcal}{OMS}{cmsy}{m}{n}

\makeatother

\title{Hamiltonian Decoded Quantum Interferometry  for General Pauli Hamiltonians}

\author{Kaifeng Bu$^{1,2}$\,}
\email{ bu.115@osu.edu (K.Bu, corresponding author)}
 \address{$^1$\textnormal{Department of Mathematics, The Ohio State University, Columbus, Ohio 43210, USA}}
\address{$^2$\textnormal{Department of Physics, Harvard University, Cambridge, Massachusetts 02138, USA}}

\author{Weichen Gu$^1$}

\author{Xiang Li$^1$}

\begin{document}

\begin{abstract}

We study Hamiltonian Decoded Quantum Interferometry (HDQI) for general Pauli Hamiltonians $H = \sum_i c_i P_i$ on an $n$-qubit system, 
where $c_i \in \mathbb{R}$ and the $P_i$ are Pauli operators.
Given access to a suitable decoding oracle, we construct efficient quantum algorithms that prepare the state $\rho_{\mathcal P}(H) = \frac{\mathcal P^2(H)}{\text{Tr}[\mathcal P^2(H)]}$ for any univariate polynomial $\mathcal P(x)$, 
which in turn approximates the Gibbs state of $H$ for appropriate choices of $\mathcal P$. 
The central technical ingredient is a matrix product state (MPS) representation of the 
reference state, which enables handling the asymmetric coefficient structure that 
arises for general real coefficients.
Our results substantially extend the scope of HDQI beyond stabilizer-like Hamiltonians, providing a unified method for Gibbs state preparation and Hamiltonian optimization across a broad class of physically and computationally relevant quantum systems.

\end{abstract}

\maketitle

\setcounter{tocdepth}{1}
\tableofcontents

\section{Introduction}

Quantum optimization—the task of using quantum algorithms to identify optimal or near-optimal solutions from an exponentially large space of feasible configurations—has emerged as one of the most promising pathways toward achieving practical quantum advantage \cite{abbas2024challenges,leng2025sub,pirnay2024principle,huang2025vast}. 
Several families of quantum optimization algorithms have been extensively investigated over the past two decades. These include amplitude-amplification-based methods such as Grover's algorithm \cite{grover1996fast}, adiabatic and annealing-based approaches \cite{farhi2000quantum,albash2018adiabatic}, 
and variational methods such as the Quantum Approximate Optimization Algorithm (QAOA) and its many generalizations \cite{farhi2014quantum,ZhouQAOA20,herrman2022multi,vijendran2023expressive,shi2022multiangle,zhao2025symmetry}. 
Although these methods have shown encouraging empirical performance, their theoretical guarantees remain limited, and in many regimes classical algorithms can match or even surpass their performance. 
Understanding when quantum optimization algorithms provide a provable advantage therefore remains a central open problem.

Decoded Quantum Interferometry (DQI), recently introduced by Jordan et al. ~\cite{jordan2024optimization}, represents a novel non-variational paradigm for quantum optimization. By employing quantum interference as its primary resource, DQI utilizes the quantum Fourier transform (QFT) and a classical decoding algorithm to concentrate probability amplitudes on bitstrings associated with high objective values, thereby enhancing the probability of sampling high-quality solutions.  These characteristics suggest a scalable framework with the potential to provide exponential speedups for specific problem classes~\cite{jordan2024optimization,khattar2025verifiable}. 
The computational complexity of DQI has been analyzed in several settings \cite{marwaha2025complexity,sabater2025towards,Parekh2025no,Anschuetz2025decoded,wang2025kernelized,rosmanis2026nearly}, and its robustness to noise and imperfect operations has also begun to be explored \cite{bu2025decoded}, highlighting its promise for near-term and future quantum devices.

An important extension of this paradigm is Hamiltonian Decoded Quantum Interferometry (HDQI), 
introduced by Schmidhuber et al.~\cite{schmidhuber2025hamiltonian}, 
which adapts the DQI framework from classical combinatorial optimization to the quantum setting of Hamiltonian optimization and Gibbs state preparation. 
Gibbs states play a central role in many applications, including free-energy estimation, sampling, and statistical inference~\cite{tuckerman2023statistical,mackay2003information,landau2021guide}, and a variety of quantum algorithms for Gibbs state preparation have been developed, 
ranging from quantum Metropolis sampling~\cite{temme2011quantum} to recent constructions based on Lindbladian dynamics~\cite{chen2023quantum,chen2023efficient,ding2024polynomial}.

HDQI provides a fundamentally different route: it employs coherent Bell measurements together with the symplectic representation of the Pauli group to reduce Gibbs sampling and energy minimization to a classical decoding problem~\cite{schmidhuber2025hamiltonian}.
This approach enables efficient preparation of Gibbs states at arbitrary temperatures for a class of physically motivated commuting Hamiltonians, 
especially those arising from some well-known quantum codes. 
In these models, the Hamiltonian takes the restricted form
$H=\sum_i v_i P_i$, where
$v_i\in \set{ -1, +1}$, and $P_i$ are Pauli operators. While this structure admits efficient decoding and elegant analytical control, it significantly limits the class of Hamiltonians that can be treated.

In this work, we study the HDQI framework for general Pauli Hamiltonians of the form
$H=\sum_i c_i P_i$,  where 
$c_i\in \real$, and $P_i$ are $n$-qubit Pauli operators. 
In the original HDQI framework, the restriction $v_i\in\{\pm1\}$ is not merely a notational simplification: it leads to a highly symmetric coefficient structure in the polynomial expansion of $\CPP(H)$. 
For general real coefficients $c_i$, the coefficients become weighted multinomial sums depending on the full coefficient profile of the Hamiltonian. 
A central technical contribution of this work is to show that these weighted coefficients can still be coherently encoded in efficiently preparable tensor-network reference states.
This work not only dramatically broadens the scope of HDQI to encompass a wide range of physically and computationally relevant models, but also provides a new method for efficiently preparing the corresponding reference states.

\begingroup
\setlength{\tabcolsep}{6pt} 
\renewcommand{\arraystretch}{1.3} 

\begin{table*}
\centering
 \resizebox{\textwidth}{!}{
\centering
\begin{tabular}{ lccc }
\toprule
& \textbf{Commuting case} & \textbf{Nearly independent commuting case} & \textbf{Noncommuting case}\\[-3pt]
& (Theorem \ref{thm:ref-state}) & (Theorem \ref{thm:ref-state-nearly-indep}) & (Theorem \ref{thm:ref-state-noncomm}) \\
\midrule
Reference state & $|R^l(H)\rangle$  & $\ket{R^l_k(H)}$ & $\ket{R^l_*(H)}$  \\

 Bond dimension $D$ & $l+1$ &  $2^k (l+1)$  & $l+1$ \\

Classical pre-processing time & $O(ml^3)$ & $O(2^{2k}\cdot m\cdot l^3)$ & $O(m\cdot (l+\CMM)^{O(\CMM)})$ \\

Quantum preparation time  & $O(m\cdot \mathrm{poly}(l))$   & $O(m\cdot \mathrm{poly}(2^k,l) )$  & $O(m\cdot 2^\CMM\cdot\mathrm{poly}(l)  )$\\
\bottomrule
\end{tabular}}
\vskip 8pt
\caption{Summary of reference state preparation complexities across different Hamiltonian regimes. Here $m$ is the number of Pauli terms, $l$ is the polynomial degree with $l\ge1$, $k$ is the dimension of the symplectic code (number of logical bits), and $\mathcal{M}$ is the maximum size of connected components in the anticommutation graph.}
\label{sum:tab1}
\end{table*}
\endgroup

\subsection{Summary of main results}
Our contributions in this work can be grouped into three parts:
(1) commuting Hamiltonians,
(2) nearly independent commuting Hamiltonians, and
(3) noncommuting Hamiltonians. Specifically:

\textbf{(1) Commuting Hamiltonians (Sec.~\ref{sec:com}):} We investigate commuting Pauli Hamiltonians of the form $H=\sum_i c_i P_i$, where all Pauli operators $\{P_i\}$ commute pairwise. Assuming access to a suitable decoding oracle, we demonstrate that for any univariate polynomial $\mathcal{P}(x)$ of degree $l$, there exists an efficient quantum algorithm to prepare the state $\rho_{\mathcal{P}}(H) = \mathcal{P}^2(H) / \mathrm{Tr}[\mathcal{P}^2(H)]$ (Theorem~\ref{thm:hdqi-commuting}). This algorithm utilizes a reference state $|R^l(H)\rangle$, defined in Eq.~\eqref{eq:ref-state}.
The key technical step is to show that $|R^l(H)\rangle$ admits an MPS representation with bond dimension $D=l+1$, which enables efficient preparation of the reference state.

\textbf{(2) Nearly Independent Commuting Hamiltonians (Sec.~\ref{sec:comnear}):} We establish that when the Pauli operators $\{P_i\}$ in $H$ are linearly independent, the required decoding oracle is guaranteed to exist. Consequently, an efficient quantum algorithm for preparing $\rho_{\mathcal{P}}(H)$ is always available (Theorem~\ref{thm:hdqi-comp_indep}). We extend these results to the regime of nearly independent Pauli operators, as detailed in Sec.~\ref{sec:comnear}.

\textbf{(3) Noncommuting Hamiltonians (Sec.~\ref{sec:noncom}):} We extend our analysis to general Pauli Hamiltonians comprising noncommuting terms. In this setting, we introduce a novel reference state $|R^l_*(H)\rangle$, defined in Eq.~\eqref{eq:ref-state-star}, and demonstrate that it also admits an MPS representation.
We show that the complexity of preparing this state is governed by the structure of the Hamiltonian's anticommutation graph  (see Table \ref{sum:tab1} for a summary of reference state preparation complexities across different regimes).
When the reference state \(\ket{R^l_*(H)}\) can be prepared efficiently,
this structure enables the construction of an efficient quantum algorithm for preparing
\(\rho_{\CPP}(H)\) for any degree-\(l\) polynomial \(\CPP(x)\), provided access to a suitable decoding oracle. This result is established in Theorem~\ref{thm:hdqi-noncomm}.  Finally, we provide some explicit examples to illustrate the utility of our framework for non-commuting systems.

\subsection{Organization of the paper}
\label{subsec:organization}

The remainder of this paper is organized as follows. 
In Sec.~\ref{sec:prelim}, we introduce the necessary background. 
In Sec.~\ref{sec:com}, we develop the HDQI framework for commuting 
Hamiltonians with arbitrary real coefficients.
Then, we extend it to the nearly-independent case, in which the 
symplectic code of the Hamiltonian has dimension $k$ in Sec.~\ref{sec:comnear}. 
In Sec.~\ref{sec:noncom}, we treat general noncommuting Pauli 
Hamiltonians, where the algorithm's complexity is governed by the structure 
of the anticommutation graph. 
In Sec.~\ref{sec6}, we discuss the effect of imperfect decoding for commuting and noncommuting Hamiltonians.
We conclude in Sec.~\ref{sec:conclusion} with a discussion of future directions. 

\section{Preliminary}
\label{sec:prelim}

In this work, we focus on an $n$-qubit system $\mathcal{H}^{\ot n}$ where $\mathcal{H}\simeq \complex^2$. For a single qubit, we consider the computational basis
$\set{\ket{0}, \ket{1}}$. The Pauli $X$ and $Z$ operators are defined by
\begin{equation}
X=\left[
\begin{array}{cc}
0&1\\
1&0
\end{array}
\right],\quad
Z=\left[
\begin{array}{cc}
1&0\\
0&-1
\end{array}
\right].
\end{equation}
The single-qubit Pauli operators are given by
\begin{align}
W(\alpha,\beta)=i^{-\alpha\beta}Z^{\alpha}X^{\beta},~~
\text{where}~~
\alpha,\beta\in\mathbb{F}_2=\set{0,1}.
\end{align}
More generally, the $n$-qubit Pauli operators are tensor products of single-qubit Paulis:
\begin{align}
W(\bm{\alpha}, \bm{\beta})=W(\alpha_1,\beta_1)\ot W(\alpha_2, \beta_2)\ot\ldots \ot W(\alpha_n, \beta_n),
\end{align}
where $\bm{\alpha}=(\alpha_1,\ldots,\alpha_n)\in \mathbb{F}^n_2$ and 
$\bm{\beta}=(\beta_1,\ldots,\beta_n)\in \mathbb{F}^n_2$. These operators satisfy the commutation relation
\begin{align}
W(\bm \alpha, \bm \beta)
W(\bm \alpha', \bm\beta')
=(-1)^{\inner{(\bm\alpha,\bm \beta)}{(\bm \alpha', \bm \beta')}_s}W(\bm \alpha', \bm \beta')W(\bm \alpha, \bm \beta),
\end{align}
where the symplectic inner product  is defined as 
$$\inner{(\bm\alpha,\bm \beta)}{(\bm \alpha', \bm \beta')}_s=\bm\alpha \cdot\bm\beta'-\bm\beta\cdot\bm\alpha'.$$ 
Accordingly, for any 
$n$-qubit Pauli operator $P$
, the vector $(\bm \alpha, \bm\beta)\in \BFF_2^{2n}$  is called its symplectic representation, denoted by $\mathrm{symp}(P )$.

Given a Hamiltonian $H$ acting on an $n$-qubit system, it can be written in the Pauli expansion
\begin{align}
H=\sum^m_{i=1}c_iP_i,
\end{align}
where $c_i\in\real $, and $P_i$ are distinct Pauli operators.  
In this work, we extend the notion of the symplectic representation from individual Pauli operators to the Hamiltonian as a whole.
Specifically, the symplectic representation of the Hamiltonian $H$, denoted by 
 $\mathrm{symp}(H)$, is defined by  the binary matrix $B^{\top} \in \BFF_2^{2n \times m}$ as follows
\begin{align}\label{eq:symH}
        B^{\top} = \begin{bmatrix}
        | & | &        & | \\
        \mathrm{symp}(P_1) & \mathrm{symp}(P_2) & \cdots & \mathrm{symp}(P_m) \\
        | & | &        & | 
        \end{bmatrix}, 
    \end{align}
 where the columns are given by the symplectic representations of the Pauli terms appearing in $H$. 
 This matrix can be naturally interpreted as the 
parity-check matrix of a classical linear code, and accordingly, we
refer to $\mathcal C^{\perp}$ as the \emph{symplectic code} associated with $H$
~\cite{jordan2024optimization, schmidhuber2025hamiltonian}. The 
associated code is
\begin{align}\label{eq:code_H}
    \mathcal{C}^{\perp} = \bigl\{ \mathbf{y} : 
    B^{\top} \mathbf{y} = \mathbf{0} \bigr\},
\end{align}
and any vector $\mathbf{y}$ with $B^{\top} \mathbf{y} \neq \mathbf{0}$ 
produces a nonzero \emph{syndrome} $B^{\top} \mathbf{y}$, from which we 
wish to recover $\mathbf{y}$.
This recovery task corresponds to the standard syndrome-decoding problem in classical coding theory. For use in our quantum circuits, we require a coherent implementation of the decoder; we refer to this coherent implementation as a \emph{decoding oracle}.

The intuition behind HDQI's reliance on a classical decoder is that the symplectic representation provides a phase-space encoding of the Hamiltonian: each Pauli operator is mapped to a vector in $\mathbb{F}_2^{2n}$, and products of Pauli operators correspond to $\mathbb{F}_2$-linear combinations of these vectors. 
Operator-level manipulations of $H$ are thereby translated into linear-algebraic operations on bitstrings, placing them in exactly the setting where classical decoding algorithms operate.

Now, let us give a formal definition of the
decoding oracle. 
Given an integer $l \ge 0$ and the matrix $B^\top$ from 
Eq.~\eqref{eq:symH} with the corresponding symplectic code $\mathcal C^\perp$ in \eqref{eq:code_H},  a \emph{weight-$l$ decoding oracle} 
$\mathcal{D}_H^{(l)}$ for $H$ is a unitary operator such that
\begin{equation}
\mathcal{D}_H^{(l)}\, \ket{\mathbf{y}} \ket{B^\top \mathbf{y}} 
\;=\; \ket{\mathbf{0}} \ket{B^\top \mathbf{y}},
\label{eq:decoding-oracle}
\end{equation}
for every $\mathbf{y} \in \mathbb{F}_2^m$ with $ |\mathbf{y}| \leq l$,
where $|\mathbf{y}|$ denotes the Hamming weight of $\mathbf{y}$. 
Implementing $\mathcal{D}_H^{(l)}$ reduces to running a classical 
syndrome decoder in superposition, followed by coherent XOR onto the 
auxiliary register; in particular, an efficient classical decoder yields 
an efficient quantum oracle. In addition, the existence and efficiency of the classical decoder for 
the code in \eqref{eq:code_H} depend on the Hamiltonian $H$.

\begin{Rem}

For a classical decoder for the code $\mathcal C^\bot$ in \eqref{eq:code_H} defined by the symplectic representation of the Hamiltonian $H$, $\mathcal{D}_H^{(l)}$, 
which is a weight-$l$ decoding oracle for $H$, can be implemented as follows.
First, the classical decoder is applied to recover the vector $\mathbf{y}$ from the syndrome $B^{\top}\mathbf{y}$. Subsequently, $\mathbf{y}$ is coherently added to its register via bitwise XOR gates, mapping the state to $\ket{\mathbf{0}}$.
Thus, implementing $\mathcal{D}_H^{(l)}$ reduces to the existence of an efficient classical decoder for the symplectic code $\mathrm{symp}(H)$.

\emph{Information-theoretically}, such a decoder exists whenever $2l < d^{\perp}$, where $d^{\perp}$ denotes the minimum distance of $\mathcal{C}^{\bot}$, since this condition guarantees that distinct error vectors of weight at most $l$ produce distinct syndromes. 
\emph{Computationally}, however, existence does not imply efficiency. 
Let us give two important regimes that admit efficient decoders.  
First, when $\mathcal{C}^{\perp}$ is a trivial code (i.e., 
$\dim \ker(\mathrm{symp}(H)) = 0$), the columns of $B^{\top}$ are linearly independent, and $\mathbf{y}$ is recovered from $B^{\top} \mathbf{y}$ by Gaussian or Gauss--Jordan elimination, correcting errors of arbitrary weight. 
Second, many well-studied classical code families, 
including Reed--Solomon and Hamming codes, admit polynomial-time decoding algorithms suitable for this purpose.
\end{Rem}

\begin{Rem}
Note that, in classical communication, one typically wants the code
$\mathcal C^\perp=\ker(B^\top)$
to be large, i.e., to have a high rate, so that many messages can be encoded while retaining redundancy for error correction. 
Accordingly, the parity-check matrix is usually chosen to be far from full rank.
In HDQI, by contrast, it is advantageous for \(B^\top\) to be full rank, or close to full rank, so that the syndrome map $\mathbf y \mapsto B^\top \mathbf y$ is injective, or nearly injective, and the vector \(\mathbf y\) can be uniquely recovered in polynomial time. 
Thus, a code that is ``trivial'' from the communication-theoretic perspective, namely one with \(\dim \mathcal C^\perp=0\), is precisely the ideal case for HDQI: in this regime, an efficient decoder exists and can be implemented using Gaussian or Gauss--Jordan elimination.
More generally, however, the classical code in HDQI is determined by the structure of the Hamiltonian, and the existence and efficiency of decoders depend on the particular Hamiltonian under consideration. 
The nearly-independent Hamiltonians studied in Sec.~\ref{sec:comnear} interpolate between the HDQI-ideal extreme \(k=0\) and genuinely nontrivial classical codes with large \(k\). 
In this intermediate regime, the associated symplectic codes have small but nonzero dimension \(k\).

\end{Rem}

Throughout the state-preparation results, $\mathcal P(x)=\sum_{j=0}^l a_jx^j$ is a real polynomial, and we assume $\mathcal P(H)\ne0$ whenever a normalized target or reference state is considered.
Given a Hamiltonian $H$ and such a polynomial $\mathcal P$, the objective of HDQI is to prepare the state $\rho_{\mathcal P}(H)=\frac{\mathcal P^2(H)}{\mathrm{Tr}[\mathcal P^2(H)]}$ using the above decoding oracle.
Here, $\mathcal P(H)=\sum_{j=0}^l a_jH^j$ is defined by polynomial functional calculus.
Since $\mathcal P(H)$ is Hermitian, $\mathcal P^2(H)$ is positive and has a nonzero trace under our assumption.
Our complexity bounds use an idealized model with exact real arithmetic, standard arithmetic-operation costs for matrix factorizations, and exact gates specified by the computed real parameters.
A finite-precision implementation prepares an approximate state and requires an additional analysis of the desired accuracy.

Such a state $\rho_{\mathcal{P}}(H)$ can approximate the Gibbs state $e^{-\beta H}$ provided that $\mathcal{P}(x)$ is a sufficiently accurate approximation of the exponential function $e^{-\beta x/2}$ (see Theorem 3 in Ref.~\cite{schmidhuber2025hamiltonian}). Motivated by this correspondence, we investigate HDQI for general Pauli Hamiltonians of the form $H = \sum_{i=1}^m c_i P_i$, where $c_i \in \mathbb{R}$ and $P_i$ are Pauli operators.

We close this section by recalling the definition of matrix product states 
(MPS), which serve as a central tool for representing the reference states 
constructed in Sections~\ref{sec:com}--\ref{sec:noncom}.

\begin{Def}
  Let $[n] = S_1\cup \cdots \cup S_r$ be a partition of $[n]$ into $r$ disjoint sets.
Any  $\mathbf{y} \in \BFF_2^n$ can then be written as $\mathbf{y} = (\mathbf{y}^{(1)}, \mathbf{y}^{(2)} ,..., \mathbf{y}^{(r)} )$, 
where $\mathbf{y}^{(i)}\in \mathbb{F}^{S_i}_2$.
A pure state $\ket{\psi}$
is called a 
$q$-ary matrix product state (MPS) with open boundary conditions and bond dimension 
$D$ if there exist matrix-valued functions
$A^{(i)}(\mathbf{y}^{(i)})\in \CN^{D\times D}$
and  boundary vectors $\mathbf{v}_L$, $\mathbf{v}_R\in \CN^D$, such that 
\[ \ket{\psi}
= \sum_{\mathbf{y}}\mathbf{v}^{\top}_L
A^{(1)}(\mathbf{y}^{(1)})\cdots A^{(r)}(\mathbf{y}^{(r)}) \mathbf{v}_R\,
|\mathbf{y}\rangle.
\]
Here $q = \max_i 2^{|S_i|}$ denotes the local physical dimension.

\end{Def}

\section{Commuting Hamiltonians}\label{sec:com}
In this section, we consider the commuting 
Hamiltonian $H=\sum_ic_iP_i$ where the coefficients $c_i\in \real$ and the Pauli operators $P_i, P_j$ commute with each other.

We begin with an algebraic identity that rewrites the operator $\CPP(H)$ in a 
form suitable for constructing the reference state. The identity exploits the 
commutativity of the Pauli operators $\{P_i\}$ and the relation $P_i^2 = I$.

We first establish the identity for formal commuting variables, then specialize 
to Pauli operators.

\begin{lem}[Formal variable version]\label{lem:formal-expansion}
Let $x = \sum_{i=1}^{m} c_i z_i$, where $c_i \in \mathbb{R}$ and the variables 
$z_i$ satisfy $z_i^2 = 1$ and $[z_i, z_j] = 0$ for all $i, j$. Then for any 
polynomial $\CPP(x) = \sum_{j=0}^{l} a_j x^j$ of degree $l$,
\begin{equation}\label{eq:formal-expansion}
\CPP(x) = \sum_{\mathbf{y} \in \mathbb{F}_2^m} w_{\mathbf{y}}\, z_1^{y_1} \cdots z_m^{y_m},
\end{equation}
where $w_{\mathbf{y}}$ is defined as
\begin{equation}\label{eq:coeff}
w_{\mathbf{y}} := \sum_{j=0}^{l} a_j \cdot j!
  \sum_{\substack{\boldsymbol{\mu} \in \mathbb{Z}_{\ge 0}^m \\ |\boldsymbol{\mu}| = j,\, \boldsymbol{\mu} \equiv \mathbf{y} \pmod{2}}} \frac{\mathbf{c}^{\boldsymbol{\mu}}}{\boldsymbol{\mu}!},
\end{equation}
with $\boldsymbol{\mu} = (\mu_1,...,\mu_m)\in \Z_{\ge 0}^m$, $|\boldsymbol{\mu} |:= \sum_{i=1}^m \mu_i$, $\mathbf{c}^{\bmu} = \prod_{i=1}^m c_i^{\mu_i}$
and  $\bmu!=\mu_1!\mu_2!...\mu_m!$.
\end{lem}

\begin{proof}
By the multinomial theorem,
\begin{equation*}
\CPP(x) = \sum_{j=0}^{l} a_j \left( \sum_{i=1}^{m} c_i z_i \right)^j 
= \sum_{j=0}^{l} a_j \sum_{\substack{\boldsymbol{\mu} \in \mathbb{Z}_{\ge 0}^m \\ |\boldsymbol{\mu}| = j}} \binom{j}{\boldsymbol{\mu}} \mathbf{c}^{\boldsymbol{\mu}}\, z_1^{\mu_1} \cdots z_m^{\mu_m},
\end{equation*}
where $\binom{j}{\boldsymbol{\mu}} = j! / (\mu_1! \cdots \mu_m!)$ and 
$\mathbf{c}^{\boldsymbol{\mu}} = \prod_i c_i^{\mu_i}$.

Since $z_i^2 = 1$, each monomial $z_1^{\mu_1} \cdots z_m^{\mu_m}$ depends only 
on the parity of $\boldsymbol{\mu}$: writing $\mathbf{y} := \boldsymbol{\mu} \pmod{2}$, 
we have $z_1^{\mu_1} \cdots z_m^{\mu_m} = z_1^{y_1} \cdots z_m^{y_m}$. Grouping 
terms by $\mathbf{y}$ yields
\begin{equation*}
\CPP(x) = \sum_{\mathbf{y} \in \mathbb{F}_2^m} 
\left( \sum_{j=0}^{l} a_j \cdot j! \sum_{\substack{\boldsymbol{\mu} \in \mathbb{Z}_{\ge 0}^m \\ |\boldsymbol{\mu}| = j,\, \boldsymbol{\mu} \equiv \mathbf{y} \pmod{2}}} \frac{\mathbf{c}^{\boldsymbol{\mu}}}{\boldsymbol{\mu}!} \right) z_1^{y_1} \cdots z_m^{y_m},
\end{equation*}
which is exactly Eq.~\eqref{eq:formal-expansion}.

Finally, $w_{\mathbf{y}} = 0$ when $|\mathbf{y}| > l$ because each 
contributing $\boldsymbol{\mu}$ satisfies $|\boldsymbol{\mu}| \ge |\mathbf{y}|$ 
(since $\mu_i \ge y_i$), so $|\boldsymbol{\mu}| > l$ exceeds the polynomial degree.
\end{proof}

\begin{lem}\label{lem:poly-expansion}
Let $\CPP(x) = \sum_{j=0}^{l} a_j x^j$ be a univariate polynomial of degree $l$, 
and let $H = \sum_{i=1}^{m} c_i P_i$ be a Hamiltonian with $c_i \in \mathbb{R}$ 
and mutually commuting Pauli operators $\{P_i\}$. Then
\begin{equation}\label{eq:PH-pauli-expansion}
\CPP(H) = \sum_{\mathbf{y} \in \mathbb{F}_2^m} w_{\mathbf{y}}\, P_{\mathbf{y}},
\end{equation}
where $P_{\mathbf{y}} := \prod_{i \in \mathrm{supp}(\mathbf{y})} P_i$ and the coefficient $w_{\mathbf y}$ is defined in Eq.~\eqref{eq:coeff}.
We call $w_{\mathbf y}$ the coefficient associated with the parity label $\mathbf y$.
If different labels give the same Pauli product up to a phase, their contributions must be combined to obtain the coefficient of that Pauli operator.
In particular, $w_{\mathbf{y}}= 0$ whenever $|\mathbf{y}| > l$.
\end{lem}

\begin{proof}
The Pauli operators $\{P_i\}$ satisfy the same algebraic relations as the formal 
variables $\{z_i\}$ in Lemma~\ref{lem:formal-expansion}: pairwise commutativity 
$[P_i, P_j] = 0$ and $P_i^2 = I$. Substituting $z_i \to P_i$ in 
Eq.~\eqref{eq:formal-expansion} therefore preserves all algebraic identities used 
in the derivation, and we obtain
\begin{equation*}
\CPP(H) = \sum_{\mathbf{y} \in \mathbb{F}_2^m} w_{\mathbf{y}}\, P_1^{y_1} \cdots P_m^{y_m}
= \sum_{\mathbf{y} \in \mathbb{F}_2^m} w_{\mathbf{y}}\, P_{\mathbf{y}},
\end{equation*}
where $P_{\mathbf{y}} = \prod_{i \in \mathrm{supp}(\mathbf{y})} P_i$ is well-defined 
as an unordered product since the $P_i$ commute.
\end{proof}

We now construct a reference state for the preparation of $\rho_{\mathcal{P}}(H)$. Intuitively, this state coherently encodes the Pauli strings comprising the polynomial expansion of $\mathcal{P}(H)$, while the bond dimension tracks the cumulative degree of the associated monomials.

We first establish a key algebraic property of the shift matrices.

\begin{lem}[Shift matrix multiplication]\label{lem:shift-mult}
For any integer $l\ge 0$, the shift matrices $B_j\in\mathbb{R}^{(l+1)\times(l+1)}$ defined by 
$(B_j)_{i, i'} = \delta_{i' - i,\, j}$ for $0\le i,i'\le l$ satisfy
\begin{equation}\label{eq:shift-product}
B_j B_k = 
\begin{cases}
B_{j+k}, & j + k \le l, \\
0, & j + k > l,
\end{cases}
\qquad \forall\, j, k \in \{0, 1, \ldots, l\}.
\end{equation}
In particular, $B_0 = I$, and $B_j$ is strictly upper-triangular for $j \ge 1$.
\end{lem}

\begin{proof}
Direct computation gives
\begin{equation*}
(B_j B_k)_{i, i''} 
= \sum_{i' = 0}^{l} (B_j)_{i, i'} (B_k)_{i', i''}
= \sum_{i' = 0}^{l} \delta_{i' - i,\, j}\, \delta_{i'' - i',\, k}
= \delta_{i'' - i,\, j + k},
\end{equation*}
which equals $(B_{j+k})_{i, i''}$ when $j + k \le l$ and vanishes otherwise 
(since $i'' \le l$ forces $i'' - i \le l$).
\end{proof}

\begin{thm}[Reference state: MPS structure and efficient preparation]
\label{thm:ref-state}
Let $H=\sum_{i=1}^{m}c_iP_i$ be a commuting Hamiltonian with $c_i\in\mathbb R$.
Let $\CPP(x)=\sum_{j=0}^{l}a_jx^j$ be a real polynomial of degree $l$ with $\CPP(H)\ne0$.
Consider the $m$-qubit state
\begin{equation}\label{eq:ref-state}
|R^l(H)\rangle 
= \frac{1}{\mathcal{N}} \sum_{\mathbf{y} \in \mathbb{F}_2^m} w_{\mathbf{y}}\, |\mathbf{y}\rangle,
\end{equation}
where $\mathcal N$ is a normalization factor and $w_{\mathbf y}$ is the coefficient associated with the parity label $\mathbf y$ in Lemma~\ref{lem:poly-expansion}.
Then:
\begin{enumerate}[(i)]
\item \textbf{(MPS structure.)} $|R^l(H)\rangle$ is a 2-ary matrix product state 
with open boundary conditions and bond dimension $D = l + 1$. Specifically,
\begin{equation}\label{eq:mps-form}
|R^l(H)\rangle 
= \sum_{\mathbf{y} \in \mathbb{F}_2^m} 
v_L^\top A^{(1)}(y_1) \cdots A^{(m)}(y_m)\, v_R\, |\mathbf{y}\rangle,
\end{equation}
with boundary vectors
\begin{equation}\label{eq:boundary-vectors} 
v_L = \tfrac{1}{\mathcal{N}}\, e_1, \quad
v_R = \bigl(a_0 \cdot 0!,\ a_1 \cdot 1!,\ \ldots,\ a_l \cdot l!\bigr)^\top,
\end{equation}
where $e_1 = (1, 0, \ldots, 0)^\top \in \mathbb{R}^{l+1}$, and site matrices
\begin{equation}\label{eq:site-matrix}
A^{(k)}(y_k) 
= \sum_{\substack{0 \le j \le l \\ j \equiv y_k \pmod{2}}} \frac{c_k^{\,j}}{j!}\, B_j,
\qquad y_k \in \{0, 1\},
\end{equation}
where $B_j \in \mathbb{R}^{(l+1) \times (l+1)}$ is the shift matrix defined by 
$(B_j)_{i, i'} = \delta_{i' - i,\, j}$ for $0 \le i, i' \le l$.
\item \textbf{(Efficient preparation.)} $|R^l(H)\rangle$ can be prepared by a classical pre-processing step taking time $T_{\mathrm{pre}}=O(m(l+1)^3)$, followed by a quantum algorithm running in time $T_{\mathrm{prepare}}=O(m\cdot\mathrm{poly}(l))$. 
The classical runtime assumes that elementary arithmetic operations and memory 
access take $O(1)$ time.
\end{enumerate}
\end{thm}

\begin{proof}
We now prove part (i) of Theorem~\ref{thm:ref-state}. Substituting the definition 
of $A^{(k)}(y_k)$ from Eq.~\eqref{eq:site-matrix} into the product 
$A^{(1)}(y_1) \cdots A^{(m)}(y_m)$ and applying Lemma~\ref{lem:shift-mult} gives
\begin{align*}
A^{(1)}(y_1) \cdots A^{(m)}(y_m) 
=\prod_{i=1}^{m} \left( \sum_{\substack{0 \le k_i \le l \\ k_i \equiv y_i \pmod{2}}} \frac{c_i^{\,k_i}}{k_i!}\, B_{k_i} \right) 
= \sum_{j=0}^{l} \left( \sum_{\substack{k_1 + \cdots + k_m = j \\ k_i \equiv y_i \pmod{2}}} \frac{c_1^{k_1}}{k_1!} \cdots \frac{c_m^{k_m}}{k_m!} \right) B_j.
\end{align*}
Contracting with the boundary vectors and using $e_1^\top B_j\, v_R = a_j \cdot j!$ 
(valid since $B_j$ shifts $e_1$ to the $j$-th standard basis vector), we obtain
\begin{align*}
v_L^\top A^{(1)}(y_1) \cdots A^{(m)}(y_m)\, v_R
= \frac{1}{\mathcal{N}} \sum_{j=0}^{l} a_j \cdot j! 
\sum_{\substack{k_1 + \cdots + k_m = j \\ k_i \equiv y_i \pmod{2}}} 
\frac{c_1^{k_1}}{k_1!} \cdots \frac{c_m^{k_m}}{k_m!} 
=\frac{1}{\mathcal{N}}\, w_{\mathbf{y}},
\end{align*}
where the second equality follows from the definition~\eqref{eq:coeff} of 
$w_{\mathbf{y}}$ and the fact that $\mathbf{k} = (k_1, \ldots, k_m)$ 
ranges over compositions equivalent to $\mathbf{y} \pmod{2}$ summing to $j$. 
This matches the coefficient of $|\mathbf{y}\rangle$ in Eq.~\eqref{eq:ref-state}, 
proving the MPS form~\eqref{eq:mps-form}. The bond dimension is $D = l + 1$ 
since each $A^{(k)}(y_k) \in \mathbb{R}^{(l+1) \times (l+1)}$. 
For part (ii), we analyze the classical pre-processing time and quantum preparation time separately.

\medskip
\noindent\textbf{Classical pre-processing.} We compute the site matrices, the right boundary vector, and the normalization factor, and then construct the quantum circuit.

\emph{(a) Site matrices.} Each coefficient $c_k^j/j!$ in Eq.~\eqref{eq:site-matrix} is computed from the preceding coefficient in $O(1)$ arithmetic operations.
Computing these coefficients and writing the site matrices takes $O(m(l+1)^2)$ operations.

\emph{(b) Right boundary vector.} The entries of $v_R=(a_0\cdot0!,\ldots,a_l\cdot l!)^\top$ are computed in $O(l+1)$ operations by evaluating the factorials successively.

\emph{(c) Normalization factor.} Using the MPS form~\eqref{eq:mps-form}, the normalization satisfies
\begin{align*}
\mathcal N^2&=\sum_{\mathbf y\in\mathbb F_2^m}\left(e_1^\top A^{(1)}(y_1)\cdots A^{(m)}(y_m)v_R\right)^2\\
&=v_R^\top\!\Bigg(\sum_{y_m}A^{(m)}(y_m)^\top\cdots\bigg(\sum_{y_1}A^{(1)}(y_1)^\top e_1e_1^\top A^{(1)}(y_1)\bigg) \cdots A^{(m)}(y_m)\Bigg)v_R.
\end{align*}
Evaluating the nested sums from inside out takes $O(m(l+1)^3)$ operations.

\emph{(d) Circuit construction.} The matrices can be brought into canonical MPS form using matrix factorizations, and the resulting site isometries can be extended to unitaries in $O(m(l+1)^3)$ arithmetic operations.
Combining these steps, the total classical pre-processing time is $T_{\mathrm{pre}}=O(m(l+1)^3)$.

\medskip
\noindent\textbf{Quantum preparation.} A 2-ary MPS with $m$ sites, open boundary 
conditions, and bond dimension $D$ can be prepared by a quantum circuit running 
in time $O(m \cdot \mathrm{poly}(D))$~\cite{schon2005sequential, melnikov2023quantum}. 
Setting $D = l + 1$ gives $T_{\mathrm{prepare}} = O(m \cdot \mathrm{poly}(l))$. 
\end{proof}

\begin{Rem}[Explicit form of the site matrices]
\label{rem:site-matrix-explicit}
For readers who prefer an explicit matrix representation of $A^{(k)}(y_k)$ 
defined in Eq.~\eqref{eq:site-matrix}, we record the two cases below. Both 
matrices are upper-triangular with $c_k^j / j!$ on the $j$-th superdiagonal 
whenever $j \equiv y_k \pmod{2}$, and $0$ on all other superdiagonals.

For $y_k = 0$ (even superdiagonals only):
\begin{equation*}
A_0^{(k)} =
\begin{pmatrix}
1 & 0 & \dfrac{c_k^{2}}{2!} & 0 & \dfrac{c_k^{4}}{4!} & \cdots & \cdots \\[4pt]
0 & 1 & 0 & \dfrac{c_k^{2}}{2!} & 0 & \ddots & \vdots \\[4pt]
0 & 0 & 1 & 0 & \dfrac{c_k^{2}}{2!} & \ddots & \dfrac{c_k^{4}}{4!} \\[4pt]
\vdots & \vdots & \ddots & \ddots & \ddots & \ddots & 0 \\[4pt]
0 & 0 & \cdots & 0 & 1 & 0 & \dfrac{c_k^{2}}{2!} \\[4pt]
0 & 0 & \cdots & 0 & 0 & 1 & 0 \\[4pt]
0 & 0 & \cdots & 0 & 0 & 0 & 1
\end{pmatrix}_{(l+1) \times (l+1)}.
\end{equation*}

For $y_k = 1$ (odd superdiagonals only):
\begin{equation*}
A_1^{(k)} =
\begin{pmatrix}
0 & \dfrac{c_k}{1!} & 0 & \dfrac{c_k^{3}}{3!} & 0 & \cdots & \cdots \\[4pt]
0 & 0 & \dfrac{c_k}{1!} & 0 & \dfrac{c_k^{3}}{3!} & \ddots & \vdots \\[4pt]
0 & 0 & 0 & \dfrac{c_k}{1!} & 0 & \ddots & 0 \\[4pt]
\vdots & \vdots & \ddots & \ddots & \ddots & \ddots & \dfrac{c_k^{3}}{3!} \\[4pt]
0 & 0 & \cdots & 0 & 0 & 0 & 0 \\[4pt]
0 & 0 & \cdots & 0 & 0 & 0 & \dfrac{c_k}{1!} \\[4pt]
0 & 0 & \cdots & 0 & 0 & 0 & 0
\end{pmatrix}_{(l+1) \times (l+1)}.
\end{equation*}

To illustrate, for $l = 4$ the matrices take the form
\begin{equation*}
A_0^{(k)} = 
\begin{pmatrix}
1 & 0 & \frac{c_k^2}{2} & 0 & \frac{c_k^4}{24} \\
0 & 1 & 0 & \frac{c_k^2}{2} & 0 \\
0 & 0 & 1 & 0 & \frac{c_k^2}{2} \\
0 & 0 & 0 & 1 & 0 \\
0 & 0 & 0 & 0 & 1
\end{pmatrix},
\qquad
A_1^{(k)} = 
\begin{pmatrix}
0 & c_k & 0 & \frac{c_k^3}{6} & 0 \\
0 & 0 & c_k & 0 & \frac{c_k^3}{6} \\
0 & 0 & 0 & c_k & 0 \\
0 & 0 & 0 & 0 & c_k \\
0 & 0 & 0 & 0 & 0
\end{pmatrix}.
\end{equation*}

\end{Rem}

Using this efficient MPS preparation of the reference state, we are able to
prepare the target state $\rho_{\CPP}(H)$.

\begin{figure*}
    \centering
    \includegraphics[width=0.9\linewidth]{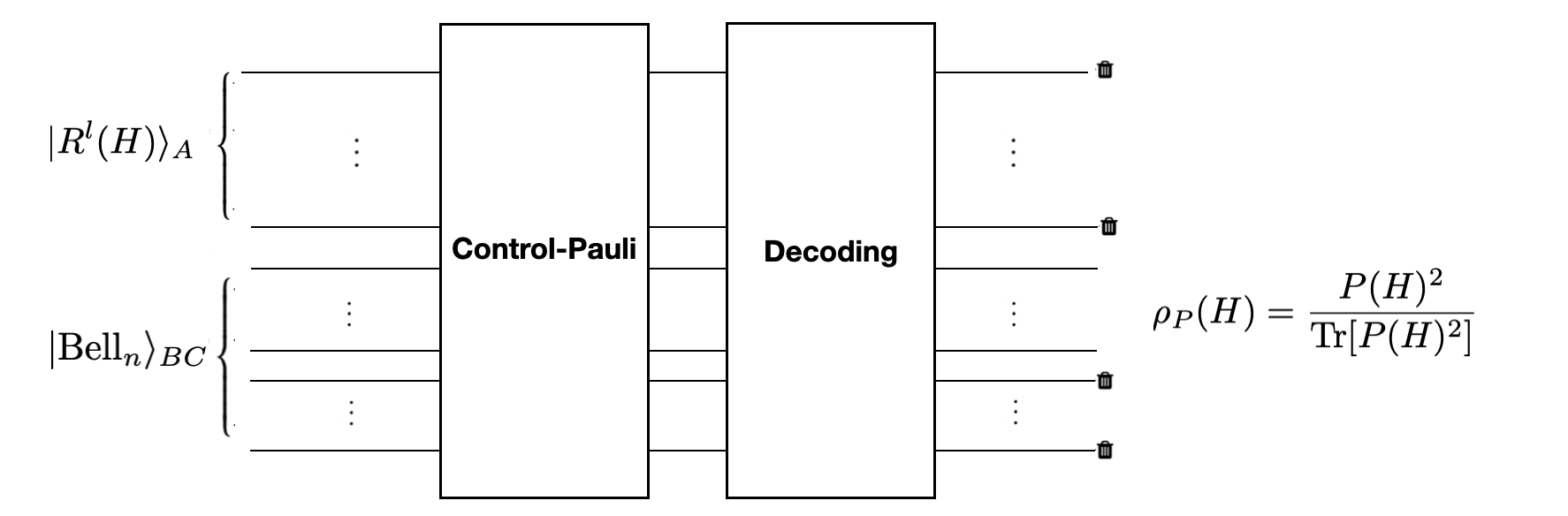}
    \caption{A quantum circuit of HDQI to generate $\rho_{\CPP}(H)$.}
    \label{fig:1}
\end{figure*}

\begin{thm}\label{thm:hdqi-commuting}
Let $H = \sum_{i=1}^{m} c_i P_i$ be a Hamiltonian on $n$ qubits, where $c_i \in \mathbb{R}$ 
and $P_1, \ldots, P_m$ are distinct, mutually commuting $n$-qubit Pauli operators. 
Assume the existence of a weight-$l$ decoding oracle $\mathcal{D}_H^{(l)}$ for $H$, as defined in Eq.~\eqref{eq:decoding-oracle}.
Let $\CPP(x)=\sum_{j=0}^{l}a_jx^j$ be a real polynomial of degree $l$ with $\CPP(H)\ne0$.
Then there exists a quantum algorithm (see Fig.~\ref{fig:1}) that prepares the state
\begin{equation}\label{eq:rho-PH}
\rho_{\CPP}(H) = \frac{\CPP^2(H)}{\mathrm{Tr}\bigl[\CPP^2(H) \bigr]}.
\end{equation}
The algorithm uses a single call to $\mathcal D_H^{(l)}$.
The classical pre-processing time is $O(m(l+1)^3)$, and the quantum runtime is $O(m\cdot\mathrm{poly}(l,n))$ plus the cost of one oracle call.
\end{thm}

\begin{proof}
By Lemma~\ref{lem:poly-expansion}, the operator $\CPP(H)$ admits the expansion
\begin{equation*}
\CPP(H) = \sum_{\mathbf{y} \in \mathbb{F}_2^m} w_{\mathbf{y}}\, P_{\mathbf{y}},
\qquad P_{\mathbf{y}} := \prod_{i \in \mathrm{supp}(\mathbf{y})} P_i,
\end{equation*}
where the coefficient $w_{\mathbf y}$  is defined as in \eqref{eq:coeff} and  $P_{\mathbf{y}}$ is well-defined as an unordered product since the $P_i$ pairwise commute. 
Note that $w_{\mathbf{y}} = 0$ whenever $|\mathbf{y}| > l$, so only Pauli 
strings of Hamming weight at most $l$ contribute.

\medskip
\noindent\textbf{Step 1: Prepare the reference state.}
On register $A$ (consisting of $m$ qubits), prepare the state $|R^l(H)\rangle$ defined in Eq.~\eqref{eq:ref-state}.
By Theorem~\ref{thm:ref-state}, this requires classical pre-processing time $O(m(l+1)^3)$ and quantum time $O(m\cdot\mathrm{poly}(l))$.

\noindent\textbf{Step 2: Prepare Bell pairs.}
On registers $B$ and $C$ (each consisting of $n$ qubits), prepare the maximally 
entangled state
\begin{equation*}
|\Phi_n\rangle = \frac{1}{\sqrt{2^n}} \sum_{\mathbf{y} \in \mathbb{F}_2^n} |\mathbf{y}\rangle \otimes |\mathbf{y}\rangle.
\end{equation*}
The joint initial state is
\begin{equation*}
|\psi_1\rangle_{ABC} = |R^l(H)\rangle_A \otimes |\Phi_n\rangle_{BC} 
= \frac{1}{\mathcal{N}} \sum_{\mathbf{y} \in \mathbb{F}_2^m} w_{\mathbf{y}}\, |\mathbf{y}\rangle_A \otimes |\Phi_n\rangle_{BC}.
\end{equation*}

\noindent\textbf{Step 3: Apply controlled Paulis.}
Apply the controlled unitary $\prod_{i=1}^{m} C^{(i)}\!P_i^{A \to B}$, where 
$C^{(i)}\!P_i^{A \to B}$ applies $P_i$ to register $B$ conditioned on the $i$-th 
qubit of register $A$ being in state $|1\rangle$. Since the $P_i$ commute, the 
order of application does not matter, and the cumulative action on register $B$ 
controlled by $|\mathbf{y}\rangle_A$ is precisely $P_{\mathbf{y}}$. The state becomes
\begin{equation*}
|\psi_2\rangle_{ABC} 
= \frac{1}{\mathcal{N}} \sum_{\mathbf{y} \in \mathbb{F}_2^m} w_{\mathbf{y}}\, |\mathbf{y}\rangle_A \otimes (P_{\mathbf{y}} \otimes I)\, |\Phi_n\rangle_{BC}.
\end{equation*}

\noindent\textbf{Step 4: Coherent Bell measurement.}
For $\mathbf u=(\bm\alpha,\bm\beta)\in\mathbb F_2^{2n}$, let $W(\mathbf u)=W(\bm\alpha,\bm\beta)$ be the Pauli operator defined in Sec.~\ref{sec:prelim}.
Let $U_{\mathrm{Bell}}$ be a coherent Bell transform satisfying
\begin{equation*}
U_{\mathrm{Bell}}(W(\mathbf u)\otimes I)|\Phi_n\rangle=|\mathbf u\rangle.
\end{equation*}
It can be implemented in $O(n)$ gates using a standard Bell transform, fixed phase corrections on each qubit pair, and a permutation of output qubits.
For each $\mathbf y$, write $P_{\mathbf y}=\omega_{\mathbf y}W(B^\top\mathbf y)$, where $\omega_{\mathbf y}\in\{1,-1,i,-i\}$.
Applying $U_{\mathrm{Bell}}$ to registers $B$ and $C$ gives
\begin{equation*}
|\psi_3\rangle_{ABC}=\frac{1}{\mathcal N}\sum_{\mathbf y\in\mathbb F_2^m}w_{\mathbf y}\omega_{\mathbf y}|\mathbf y\rangle_A\otimes|B^\top\mathbf y\rangle_{BC}.
\end{equation*}

\noindent\textbf{Step 5: Decode and discard.}
Apply the decoding oracle $\mathcal D_H^{(l)}$ to registers $A$ and $BC$.
Since $w_{\mathbf y}=0$ for $|\mathbf y|>l$, discarding register $A$ yields
\begin{equation*}
|\psi_4\rangle_{BC}=\frac{1}{\mathcal N}\sum_{\mathbf y\in\mathbb F_2^m}w_{\mathbf y}\omega_{\mathbf y}|B^\top\mathbf y\rangle_{BC}.
\end{equation*}

\noindent\textbf{Step 6: Reverse the Bell measurement.}
Applying $U_{\mathrm{Bell}}^\dagger$ and using $P_{\mathbf y}=\omega_{\mathbf y}W(B^\top\mathbf y)$ gives
\begin{align*}
|\psi_5\rangle_{BC}=\frac{1}{\mathcal N}\sum_{\mathbf y\in\mathbb F_2^m}w_{\mathbf y}(P_{\mathbf y}\otimes I)|\Phi_n\rangle_{BC}=\frac{1}{\mathcal N}(\CPP(H)\otimes I)|\Phi_n\rangle_{BC}.
\end{align*}
The last equality follows from Eq.~\eqref{eq:PH-pauli-expansion}.

\noindent\textbf{Step 7: Trace out register $C$.}
Let $\{|\lambda\rangle\}$ denote an eigenbasis of $H$ with eigenvalues $\lambda$. 
The maximally entangled state admits the Schmidt decomposition
\begin{equation*}
|\Phi_n\rangle_{BC} = \frac{1}{\sqrt{2^n}} \sum_\lambda |\lambda\rangle_B \otimes |\bar\lambda\rangle_C,
\end{equation*}
where $|\bar\lambda\rangle$ is the complex conjugate of $|\lambda\rangle$ in the 
computational basis. Hence
\begin{equation*}
(\CPP(H) \otimes I)\, |\Phi_n\rangle_{BC} = \frac{1}{\sqrt{2^n}} \sum_\lambda \CPP(\lambda)\, |\lambda\rangle_B \otimes |\bar\lambda\rangle_C.
\end{equation*}
Tracing out register $C$ yields
\begin{align*}
\Ptr{C}{\proj{\psi_5}_{BC}}
= \frac{1}{\mathcal{N}^2 \cdot 2^n} \sum_\lambda \CPP(\lambda)^2 |\lambda\rangle\langle\lambda|_B 
= \frac{\CPP^2(H)}{\mathrm{Tr}[\CPP^2(H)]} = \rho_{\CPP}(H),
\end{align*}
where the last equality follows from $\mathcal{N}^2 \cdot 2^n = \mathrm{Tr}[\CPP^2(H)]$ 
(by normalization of $|\psi_5\rangle$).

\medskip
\noindent\textbf{Total runtime.} 
The classical pre-processing for the reference state takes $O(m(l+1)^3)$ time (Theorem~\ref{thm:ref-state}). The quantum circuit consists of: 
preparing $|R^l(H)\rangle$ in time $O(m \cdot \mathrm{poly}(l))$; 
preparing $|\Phi_n\rangle$ in time $O(n)$; 
applying $m$ controlled Paulis, each acting on $n$ qubits, in time $O(mn)$; 
the coherent Bell measurement and its inverse in time $O(n)$; 
and a single call to $\mathcal{D}_H^{(l)}$. 
The total quantum runtime is $O(m \cdot \mathrm{poly}(l, n))$, 
plus the cost of one decoding oracle call.
\end{proof}

\begin{Rem}
Our algorithms are conditional on the availability of efficient classical decoders for the symplectic code. 
 This assumption is satisfied for many structured code families.
For example, when $\mathrm{symp}(H)$ is a trivial code (i.e., $\dim \ker(\mathrm{symp}(H)) = 0$), efficient classical methods—such as Gaussian elimination or Gauss–Jordan elimination—can correct errors of arbitrary weight (see Theorem \ref{thm:hdqi-comp_indep}).
\end{Rem}

\begin{cor}[Efficient Gibbs state preparation for commuting Hamiltonians]
\label{cor:gibbs-commuting}
Let $H = \sum_{i=1}^{m} c_i P_i$ be a Hamiltonian on $n$ qubits with 
$c_i \in \mathbb{R}$ and mutually commuting Pauli operators $\{P_i\}$. For 
any inverse temperature $\beta > 0$ and target accuracy $\delta \in (0, 1)$, 
there exists a univariate polynomial $\CPP(x)$ of degree
\begin{equation}\label{eq:gibbs-degree-bound}
l \le 1.12\, \beta \|H\| + 0.648 \ln \frac{24}{\delta},
\end{equation}
where $\|H\|$ denotes the operator norm of $H$, such that the state 
$\rho_{\CPP}(H) = \CPP^2(H) / \mathrm{Tr}[\CPP^2(H)]$ is $\delta$-close in trace 
distance to the Gibbs state:
\begin{equation}\label{eq:gibbs-approx-commuting}
\left\| \rho_{\CPP}(H) - \frac{\exp(-\beta H)}{\mathrm{Tr}\bigl[\exp(-\beta H)\bigr]} \right\|_1 \le \delta.
\end{equation}
Given a weight-$l$ decoding oracle $\mathcal D_H^{(l)}$, the algorithm in Theorem~\ref{thm:hdqi-commuting} prepares $\rho_\CPP(H)$ in classical pre-processing time $O(m(l+1)^3)$ and quantum time $O(m\cdot\mathrm{poly}(l,n))$, plus one oracle call.
If $m$ and $\beta\|H\|$ are polynomial in $n$ and the oracle runs in $\mathrm{poly}(n,\ln(1/\delta))$ time, the total runtime is $\mathrm{poly}(n,\ln(1/\delta))$.
\end{cor}

\begin{proof}
The existence of a polynomial $\CPP(x)$ of degree $l$ satisfying Eqs.~\eqref{eq:gibbs-degree-bound} and~\eqref{eq:gibbs-approx-commuting} is established in~\cite[Theorem~3]{schmidhuber2025hamiltonian}.
Given such a polynomial and a weight-$l$ decoding oracle, Theorem~\ref{thm:hdqi-commuting} provides an algorithm preparing $\rho_\CPP(H)$ with the stated complexity.
The stated polynomial-time scaling follows by substituting Eq.~\eqref{eq:gibbs-degree-bound} into the runtime and including the cost of the oracle.
\end{proof}


\section{Commuting Hamiltonians with Nearly Independent Pauli Operators}
 \label{sec:comnear}

In Theorem~\ref{thm:hdqi-commuting}, the reference state is supported on labels of weight at most $\min\{l,m\}$.
Thus, when $l>m$, it is enough to use a weight-$m$ decoding oracle.
We first consider Hamiltonians whose symplectic representations $\mathrm{symp}(P_i)$ are linearly independent in $\mathbb F_2^{2n}$.
In this regime, $m\leq 2n$ and 
the parity-check matrix  $B^{\top}$
associated with the symplectic code of $H$
has full rank $m$, and the syndrome $B^{\top} \mathbf{y}$
uniquely determines the vector $\mathbf{y}\in \BFF_2^m$.
 Consequently, a weight-$m$ decoding oracle $\CDD_H^{ (m)}$ exists for $H$.

\begin{thm}\label{thm:hdqi-comp_indep}
Let $H = \sum_{i=1}^m c_i P_i$ be a  Hamiltonian on $n$ qubits, where $c_i\in \R$ and $P_i$ are distinct,
mutually commuting $n$-qubit Pauli operators,
such that $\mathrm{symp}(P_1 )$, ...,  $\mathrm{symp}(P_m )$ are linearly independent in $\BFF_2^{2n}$.
Let $\CDD_H^{ (m)}$ be a weight-$m$ decoding oracle for $H$. 
Let $\CPP(x)$ be a real polynomial of degree $l$ with $\CPP(H)\ne0$, where $l$ may be larger than $m$. 
Then there is a quantum algorithm that prepares the state
\begin{align}
\rho_{\CPP}(H) = \frac{{\CPP}^2(H)}{\Tr{{\CPP}^2(H)}},
\end{align}
using a single call to $\CDD_H^{ (m)}$.
The remaining quantum runtime is $\mathrm{poly}(n,l)$, and the decoder can be implemented in $\mathrm{poly}(n)$ time by Gaussian elimination.
\end{thm}
\begin{proof}
The proof is similar to that of Theorem \ref{thm:hdqi-commuting}.
The only point to note is that 
we use a weight-$m$ decoding oracle $\CDD_H^{ (m)}$ even if $l>m$ in the step of decoding.
This is valid because in this step we only need to recover $\mathbf{y}$ from $\mathrm{symp}(P_{\mathbf{y}})$, where $\mathbf{y}$ has Hamming weight at most $m$.
\end{proof}

The above theorem states that $\rho_{\CPP}(H)$ can be prepared in polynomial time for a polynomial $\CPP(x)$ of arbitrary degree $l$,
provided that the symplectic representations of the Pauli operators in $H$ are linearly independent.
That is,
this holds when the symplectic code of $H$ has dimension zero, i.e., when the number of logical bits $k=\dim \;\ker(\mathrm{symp}(H))=0$.

Beyond this fully independent setting, there exist many physically and computationally relevant Hamiltonians for which the Pauli operators are not completely independent, yet the corresponding symplectic codes have only small dimension $k$ (see~\cite{schmidhuber2025hamiltonian}). Motivated by these examples, we now turn to the more general class of commuting Hamiltonians $H=\sum_ic_iP_i$ whose Pauli terms are nearly independent. 
Specifically, we define the logical dimension of the symplectic code  $\mathcal{C}^{\perp} = \ker(B^{\top})$ from \eqref{eq:code_H} as follows
\begin{equation}
k := \dim\ker(B^{\top})
    \in \{0, 1, \ldots, m\}.
\label{eq:k-definition}
\end{equation}
 Equivalently, $k$ is the number of independent $\mathbb{F}_2$-linear relations among the 
symplectic vectors $\mathrm{symp}(P_1), \ldots, \mathrm{symp}(P_m)$. Therefore, the Hamiltonian is called nearly independent when $k$ is small, e.g.,
$k=O(\log n)$.


To study such  nearly independent commuting Hamiltonians, 
we need the following lemma.
\begin{lem}[ Theorem 15 in \cite{schmidhuber2025hamiltonian}]\label{lem:regroup_near_indep}
Let $z_i$ be formal commuting variables satisfying $z_i^2 = 1$, $i=1,...,m$. 
Suppose that there are exactly $k$ independent relations among the variables $\{z_i\}$, i.e.\ $k$ independent subsets of the $\{z_i\}$ whose product equals the identity. 
Assume that the first $d \le m$ variables $\{z_1,\cdots, z_d\}$ form a maximally independent set. 
Then there exists a partition of these $d$ variables into $r$ blocks $V_1,\ldots,V_r$ with $r \le 2^k$, 
such that for every $j>d$,
$z_j$ can be represented as 
\[z_j = \prod_{i\in U_j} z_i,\]
where each $U_j$ is a union of full blocks $V_t$,
i.e.
$$
U_j=\bigcup_{t\in T_j} V_t
\quad\text{for some } T_j\subseteq\{1,\dots,r\}.
$$
\end{lem}

The above lemma allows us to regroup the Pauli operators in $H = \sum_{i=1}^m c_i P_i$.
Assume that the first $m-k$ Pauli operators form a maximally independent set,
i.e., their symplectic representations span a subspace of dimension $m-k$,
where $k$ is the dimension of the symplectic code of $H$.
By Lemma \ref{lem:regroup_near_indep},
there exists a partition $\{1,\dots,m-k\}=V_1\sqcup\cdots\sqcup V_r$ with $r\le 2^k$, 
such that for every $1\le j\le k$,
there is a union $U_j$ of some of the subsets $V_1,...,V_r$ such that
\begin{align}\label{251209eq3}
P_{m-k+j} =  \pm \prod_{i\in U_j} P_i.
\end{align}
Without loss of generality, 
we may assume all signs in \eqref{251209eq3} are positive,
as if for some $j$, we have Pauli $P_{m-k+j} =  - \prod_{i\in U_j} P_i$, then we can replace $P_{m-k+j}$ by $-P_{m-k+j}$ and replace $c_{m-k+j}$ by $-c_{m-k+j}$,
which does not change the Hamiltonian.
Hence, 
the Hamiltonian $H$ can be written as
\begin{align}\label{eq:H-nearly-indep} 
H = \sum_{i=1}^{m-k} c_i P_i +\sum_{j=1}^k c_{m-k+j} \prod_{i\in U_j} P_i,
\end{align}
where $P_1,...,P_{m-k}$ are distinct Pauli operators such that they commute with each other and their symplectic representations are linearly independent.

Let us start with 
the simplest case. 
Suppose the Hamiltonian $H$ can be written as
\begin{align}\label{eq:exm_nin}
H = \left( \sum_{i=1}^m c_i P_i \right)+c_{m+1}P_1P_2\cdots P_m,
\end{align}
where $c_i\in \R$ and $\{P_i\}_{i=1}^m$ are independent and  commuting $n$-qubit Pauli operators.
In this case, there is always a weight-$m$ decoding oracle $\CDD_{H'}$  for the Hamiltonian $H'=\sum_{i=1}^m c_i P_i$, 
which is the decoding oracle we will use for the Hamiltonian $H$.

\begin{lem}\label{251204lem1}
Let $x = \sum_{i=1}^m c_iz_i + c_{m+1}z_1z_2\cdots z_m$ such that $c_i\in\R$, $z^2_i = 1$ for all $i$ and
$z_iz_j = z_jz_i$ for all $i,j$.
We have that, for any $l\ge 0$,
\begin{align}\label{eqn:x^l}
x^l =& \sum_{
\substack{\mathbf{y}\in \BFF_2^m  
}} l! 
\left(
I_1^l(\mathbf{y})+I_2^l(\mathbf{y})
\right)
z_1^{y_1}\cdots z_m^{y_m},
\end{align}
where 
\begin{equation}\label{eqn:I1}
I^l_1(\mathbf{y}) = \sum_{
\substack{
j=0\\
j\text{ is even}
}}^l
\sum_{
\substack{
\boldsymbol{\mu} \in \Z_{\ge 0}^m\\
|\boldsymbol{\mu}|=l-j\\
\boldsymbol{\mu} \equiv \mathbf{y} (\bmod 2) 
}} 
\frac{\mathbf{c}^{\bmu}}{\bmu!} 
\frac{c_{m+1}^{j}}{j!},
\end{equation}
and
\begin{equation}\label{eqn:I2}
I_2^l(\mathbf{y}) = \sum_{
\substack{
j=0\\
j\text{ is odd}
}}^l
\sum_{
\substack{
\boldsymbol{\mu} \in \Z_{\ge 0}^m\\
|\boldsymbol{\mu}|=l-j\\
\boldsymbol{\mu} \equiv \mathbf{y} + \mathbf{1} (\bmod 2) 
}} 
\frac{\mathbf{c}^{\bmu}}{\bmu!} 
\frac{c_{m+1}^{j}}{j!},
\end{equation}
where $\mathbf{1}=(1,1,\cdots,1)$ denotes the all-ones vector in the summation. 
\end{lem}

\begin{proof}
First, $x^l$ can be written as 
\begin{align*}
x^l
= \left(\sum_{i=1}^m c_iz_i + c_{m+1}z_1z_2\cdots z_m\right)^l
=\sum_{j=0}^l\binom{l}{j}
\left(\sum_{i=1}^m c_iz_i\right)^{l-j}
\left(c_{m+1}z_1z_2\cdots z_m\right)^j.
\end{align*}
Let 
\begin{equation*}
C_{\mathbf{y}}^l=l!\sum_{
\substack{
\boldsymbol{\mu} \in \Z_{\ge 0}^m\\
|\boldsymbol{\mu}|=l\\
\boldsymbol{\mu} \equiv \mathbf{y} (\bmod 2) 
}} 
\frac{\mathbf{c}^{\bmu}}{\bmu!}.  
\end{equation*} 
Note that each  $z_j$ has $z_j^2 =1$. 
Using  Eq.~\eqref{eq:formal-expansion}, we can rewrite $x^l$ as follows
\begin{align*}
x^l
& = \sum_{j=0}^l\binom{l}{j} 
\left(\sum_{\mathbf{y}\in \BFF_2^m}
C_{\mathbf{y}}^{l-j}\,
z_1^{y_1} \cdots z_m^{y_m}\right)
\left( c_{m+1}^j z_1^j z_2^j\cdots z_m^j\right)\\
& = \sum_{\mathbf{y}\in \BFF_2^m} 
\left(
\sum_{\substack{j=0\\j\text{ is even}}}^l \binom{l}{j} C_{\mathbf{y}}^{l-j} c_{m+1}^j+
\sum_{\substack{j=0\\j\text{ is odd}}}^l \binom{l}{j} C_{\mathbf{y}}^{l-j} c_{m+1}^j z_1z_2\cdots z_m
\right)
z_1^{y_1} \cdots z_m^{y_m}\\
& = \sum_{\mathbf{y}\in \BFF_2^m} 
\left(
\sum_{\substack{j=0\\j\text{ is even}}}^l \binom{l}{j} C_{\mathbf{y}}^{l-j} c_{m+1}^j+
\sum_{\substack{j=0\\j\text{ is odd}}}^l \binom{l}{j} C_{\mathbf{y}+\mathbf{1}}^{l-j} c_{m+1}^j
\right)
z_1^{y_1} \cdots z_m^{y_m}.
\end{align*}
For the sum of even terms,
\begin{align*}
\sum_{\substack{j=0\\j\text{ is even}}}^l \binom{l}{j} C_{\mathbf{y}}^{l-j} c_{m+1}^j
 =& \sum_{\substack{j=0\\j\text{ is even}}}^l \frac{l!}{j!(l-j)!}\; (l-j)!\sum_{
\substack{
\boldsymbol{\mu} \in \Z_{\ge 0}^m\\
|\boldsymbol{\mu}|=l-j\\
\boldsymbol{\mu} \equiv \mathbf{y} (\bmod 2) 
}} 
\frac{\mathbf{c}^{\bmu}}{\bmu!} \; c_{m+1}^j\\
 =& l! \sum_{\substack{j=0\\j\text{ is even}}}^l \sum_{
\substack{
\boldsymbol{\mu} \in \Z_{\ge 0}^m\\
|\boldsymbol{\mu}|=l-j\\
\boldsymbol{\mu} \equiv \mathbf{y} (\bmod 2) 
}} 
\frac{\mathbf{c}^{\bmu}}{\bmu!} 
\frac{c_{m+1}^{j}}{j!}.
\end{align*}
A similar procedure can be applied to the sum of odd terms, and thus we obtain \eqref{eqn:x^l}.
\end{proof}

\begin{lem}\label{lem:251206-1}
Given a real polynomial $\CPP(x)=\sum_{j=0}^la_jx^j$ of degree $l$ and the Hamiltonian $H$ in Eq.~\eqref{eq:exm_nin}, assume $\CPP(H)\ne0$.
Let us consider the following $m$-qubit reference state:
\begin{equation}\label{eqn:Rl(H)-lindep}
\ket{R^l_1(H)}  = \frac{1}{\CNN}\sum_{
\substack{\mathbf{y}\in \BFF_2^m  
}} 
\left(
\sum_{s=0}^l a_s \cdot s! 
\left(
I_1^s(\mathbf{y})+I_2^s(\mathbf{y})
\right)
\right)
\ket{\mathbf{y}},
\end{equation}
where $\CNN$ is a normalization factor, and $I_1^s(\mathbf{y})$ and $I_2^s(\mathbf{y})$ are defined in Eqs.~\eqref{eqn:I1} and \eqref{eqn:I2}.
Then, this reference state is also a $2$-ary matrix product state with open boundary conditions and bond dimension $D = 2(l + 1)$.
More specifically,
\begin{equation}\label{eqn:RlH-nearly-indep}
\begin{split}
\ket{R^l_1(H)}&=\sum_{\mathbf y\in\mathbb F_2^m}{\mathbf v_L^{(1)}}^\top\mathcal A^{(1)}(y_1)\cdots\mathcal A^{(m)}(y_m) \mathcal A^{(m+1)}(0)\mathbf v_R^{(1)}|\mathbf y\rangle.
\end{split}   
\end{equation}
with
\begin{align*}
&\mathbf{v}^{(1)}_L=\bigl(\;
\underbrace{\frac{1}{\CNN},0,\ldots,0}_{l+1}\ ,\ 
\underbrace{\frac{1}{\CNN},0,\ldots,0}_{l+1}
\;\bigr)^{\top}, \\
&\mathbf{v}^{(1)}_R=(a_0\cdot0!,\,a_1\cdot 1!,\, \cdots,\,a_l\cdot l!,\,a_0\cdot0!,\,a_1\cdot 1!,\, \cdots,\,a_l\cdot l!)^{\top},\\
&\mathcal A^{(i)}(y_i)=\begin{pmatrix}A_{y_i}^{(i)}&\mathbf0\\\mathbf0&A_{1-y_i}^{(i)}\end{pmatrix}\quad(1\le i\le m),\\
&\mathcal A^{(m+1)}(0)=\begin{pmatrix}A_0^{(m+1)}&\mathbf0\\\mathbf0&A_1^{(m+1)}\end{pmatrix}.
\end{align*}
Here $A_0^{(i)}=A^{(i)}(0)$ and $A_1^{(i)}=A^{(i)}(1)$ are the matrices in Eq.~\eqref{eq:site-matrix}, using the coefficient $c_i$ for $1\le i\le m+1$.
\end{lem}
\begin{proof}
Let $v_L=\mathcal N^{-1}e_1$ and $v_R=(a_0\cdot0!,\ldots,a_l\cdot l!)^\top$, where $e_1=(1,0,\ldots,0)^\top\in\mathbb R^{l+1}$.
By the shift-matrix identity in Lemma~\ref{lem:shift-mult}, for each $\mathbf y\in\mathbb F_2^m$ we have
\begin{align*}
v_L^\top A_{y_1}^{(1)}\cdots A_{y_m}^{(m)}A_0^{(m+1)}v_R&=\frac1{\mathcal N}\sum_{s=0}^la_ss!\,I_1^s(\mathbf y),\\
v_L^\top A_{1-y_1}^{(1)}\cdots A_{1-y_m}^{(m)}A_1^{(m+1)}v_R&=\frac1{\mathcal N}\sum_{s=0}^la_ss!\,I_2^s(\mathbf y).
\end{align*}
The two diagonal blocks of the product in Eq.~\eqref{eqn:RlH-nearly-indep} give these two expressions, so their sum is the coefficient of $|\mathbf y\rangle$ in Eq.~\eqref{eqn:Rl(H)-lindep}.
This proves the MPS representation.
\end{proof}

The following corollary applies the preceding results to state preparation. 

\begin{cor}
Given the Hamiltonian 
$H = \left( \sum_{i=1}^m c_i P_i \right)+c_{m+1}P_1P_2\cdots P_m$,
where $c_i\in \R$ and $\{P_i\}_{i=1}^m$ are independent and  commuting $n$-qubit Pauli operators.
Let $\CDD_{H'}^{ (m)}$ be a weight-$m$ decoding oracle for the Hamiltonian $H' = \sum_{i=1}^m c_i P_i$. 
Let $\CPP(x)$ be a real polynomial of degree $l$ with $\CPP(H)\ne0$, where $l$ may be larger than $m$. 
Then there exists a quantum algorithm that prepares the state
\begin{align}
\rho_{\CPP}(H) = \frac{{\CPP}^2(H)}{\Tr{{\CPP}^2(H)}},
\end{align}
using a single call to $\CDD_{H'}^{ (m)}$.
The remaining quantum runtime is $\mathrm{poly}(n,l)$, and the decoder can be implemented in $\mathrm{poly}(n)$ time by Gaussian elimination.
\end{cor}

\begin{proof}
By Lemma~\ref{251204lem1}, the coefficients of $\CPP(H)$ in the independent Pauli products are those used in the reference state $\ket{R_{1}^l(H)}$.
By Lemma~\ref{lem:251206-1}, this state is an MPS with bond dimension $2(l+1)$, and it can be prepared efficiently as in Theorem~\ref{thm:ref-state}.
The construction in Theorem~\ref{thm:hdqi-commuting}, applied with the weight-$m$ oracle $\CDD_{H'}^{(m)}$, then prepares $\rho_{\CPP}(H)$ using one oracle call.
The running time is polynomial in $n$ and $l$ when the decoder is implemented by Gaussian elimination, since the independence of the symplectic representations gives $m\le2n$.
\end{proof}

After this simple example, let us move to the general case, i.e., the Hamiltonian 
$$H = \sum_{i=1}^{m-k} c_i P_i +\sum_{j=1}^k c_{m-k+j} \prod_{i\in U_j} P_i,$$ from 
Eq.~\eqref{eq:H-nearly-indep} .

First, we need to rewrite $\CPP(H)$ by the following lemma, which is a generalization of Lemma
\ref{lem:poly-expansion}.

\begin{lem}[Symmetric $l$-th Power Expansion]\label{lem:expansion-k1}
Let $z_1,...,z_{m-k}$ be formal commuting variables such that $z^2_j = 1$ for all $j$,
and let
\begin{align}\label{eq:x_form_k}
x = \sum_{i=1}^{m-k} c_i z_i +\sum_{j=1}^k c_{m-k+j} \prod_{i\in U_j} z_i,\quad
c_j\in\R,\; j=1,\cdots, m.
\end{align}
Then we have
\begin{align}\label{251209eq6}
x^l =& \sum_{
\substack{\mathbf{y}\in \BFF_2^{m-k}  
}}  
 \left( l!
\sum_{K\subseteq \{1,...,k\}} I_{K}^l(\mathbf{y}) 
\right)
 z_1^{y_1}\cdots z_{m-k}^{y_{m-k}} ,
\end{align}
where 
\begin{align*}
 I^l_K(\mathbf{y}) =
\sum_{
\substack{
\mathbf{s} \in \Z_{\ge 0}^{k}, |\mathbf{s}|\leq l,\\
s_j\text{ is odd if } j\in K,\\
s_j\text{ is even if } j\not\in K.
}} 
\sum_{
\substack{
\boldsymbol{\mu} \in \Z_{\ge 0}^{m-k},\\
|\boldsymbol{\mu}|=l-|\mathbf{s}|,\\
\boldsymbol{\mu} \equiv \mathbf{y}+ \chi_{K}(\bmod 2) 
}} 
\prod ^{m-k}_{i=1}\frac{c_i^{\mu_i}}{\mu_i!}\prod^k_{i=1} \frac{c_{m-k+i}^{s_i}}{s_i!}.
\end{align*} 
and
\begin{align}\label{251211eq4}
\chi_{K} \equiv \sum_{j\in K} \;\mathbf{1}_{  U_j } (\bmod 2).
\end{align}
Hence, letting $\CPP(x) = \sum_{j=0}^{l} a_j x^j$ be a univariate polynomial of degree $l$ with the variable $x$ in the form of \eqref{eq:x_form_k},
we have 
\begin{align}
\CPP(x)=\sum_{\mathbf{y} \in \mathbb{F}_2^{m-k}}\left(\sum_{t=0}^{l} a_t \cdot t! \sum_{K \subseteq \{1, \ldots, k\}} I^t_K(\mathbf{y})\right)\, \mathbf{z}^{\mathbf{y}},
\end{align}
where $\mathbf{z}^{\mathbf{y}} := \prod_{i \in \mathrm{supp}(\mathbf{y})} z_i$
\end{lem}

\begin{proof}
We have
\begin{align*}
x^l
= \left( \sum_{i=1}^{m-k} c_i z_i + \sum_{j=1}^k  c_{m-k+j} \prod_{i\in U_j} z_i   \right)^l
=  \sum_{u=0}^l \binom{l}{u}
\left(\sum_{i=1}^{m-k} c_iz_i\right)^{l-u}
\left( \sum_{j=1}^k  c_{m-k+j} \prod_{i\in U_j} z_i \right)^u.
\end{align*}
Denote
\begin{equation*}
C_{\mathbf{y}}^l=
l!
\sum_{
\substack{
\boldsymbol{\mu} \in \Z_{\ge 0}^{m-k}\\
|\boldsymbol{\mu}|=l\\
\boldsymbol{\mu} \equiv \mathbf{y} (\bmod 2) 
}} 
\frac{c_1^{\mu_1}}{\mu_1!} \cdots \frac{c_{m-k}^{\mu_{m-k}}}{\mu_{m-k}!}.  
\end{equation*} 
Note that each $z_j$ satisfies $z_j^2 =1$, it follows from \eqref{eq:formal-expansion} that 
\begin{align*}
x^l
& = \sum_{u=0}^l\binom{l}{u} 
\left(\sum_{\mathbf{y}\in \BFF_2^{m-k}}
C_{\mathbf{y}}^{l-u}\,
z_1^{y_1} \cdots z_{m-k}^{y_{m-k}}\right) 
\left( \sum_{j=1}^k  c_{m-k+j} \prod_{i\in U_j} z_i \right)^u \\
& = \sum_{u=0}^l
\binom{l}{u} 
\left(\sum_{\mathbf{y}\in \BFF_2^{m-k}}
C_{\mathbf{y}}^{l-u}\,
z_1^{y_1} \cdots z_{m-k}^{y_{m-k}}\right)  \left(\sum_{s_1+\cdots + s_k=u} 
\binom{u}{\mathbf{s}} 
\left(c_{m-k+1} \prod_{i\in U_1} z_i\right)^{s_1} \cdots 
\left(c_{m} \prod_{i\in U_k} z_i\right)^{s_k}\right)\\
& = \sum_{\mathbf{y}\in \BFF_2^{m-k}} 
\sum_{
\substack{\mathbf{s} \in \Z_{\ge 0}^{k}\\
|\mathbf{s}|\le l}}
\binom{l}{|\mathbf{s}|}
C_{\mathbf{y}}^{l-|\mathbf{s}|}
\binom{|\mathbf{s}|}{\mathbf{s}}
c_{m-k+1}^{s_1}\cdots
c_{m}^{s_k} 
\left( \prod_{i\in U_1} z_i\right)^{s_1} \cdots 
\left( \prod_{i\in U_k} z_i\right)^{s_k}
z_1^{y_1} \cdots z_{m-k}^{y_{m-k}}\\
& = \sum_{\mathbf{y}\in \BFF_2^{m-k}} 
\sum_{
\substack{\mathbf{s} \in \Z_{\ge 0}^{k}\\
|\mathbf{s}|\le l}}
l!
\sum_{
\substack{
\boldsymbol{\mu} \in \Z_{\ge 0}^{m-k}\\
|\boldsymbol{\mu}|=l- |\mathbf{s}|\\
\boldsymbol{\mu} \equiv \mathbf{y} (\bmod 2) 
}}
\frac{c_1^{\mu_1}}{\mu_1!} \cdots \frac{c_{m-k}^{\mu_{m-k}}}{\mu_{m-k}!}
\frac{c_{m-k+1}^{s_1}}{s_1!} \cdots \frac{c_{m}^{s_k}}{s_k!} 
\mathbf{z}^{\mathbf{y} + s_1 \mathbf{1}_{  U_1 }+\cdots + s_k \mathbf{1}_{  U_k } }.
\end{align*}
Note that,  for any vector $\mathbf x$,  $\mathbf{z}^{\mathbf x}:=\prod^{m-k}_{i=1}z^{x_i}_i$,  and here $\mathbf x=\mathbf{y} + s_1 \mathbf{1}_{  U_1 }+\cdots + s_k \mathbf{1}_{  U_k }$.
Then we classify the terms by the parity of all $s_t$,
and denote $K=\{t: s_t $ is odd $\}$. So we have
\begin{align*}
x^l  = 
\sum_{\mathbf{y}\in \BFF_2^{m-k}} 
l!
\sum_{K\subseteq \{1,...,k\}} 
\sum_{
\substack{
\mathbf{s} \in \Z_{\ge 0}^{k}, |\mathbf{s}|\leq l,\\
s_j\text{ is odd if } j\in K,\\
s_j\text{ is even if } j\not\in K.
}}
\sum_{
\substack{
\boldsymbol{\mu} \in \Z_{\ge 0}^{m-k}\\
|\boldsymbol{\mu}|=l- |\mathbf{s}|\\
\boldsymbol{\mu} \equiv \mathbf{y} (\bmod 2) 
}}
\frac{c_1^{\mu_1}}{\mu_1!} \cdots \frac{c_{m-k}^{\mu_{m-k}}}{\mu_{m-k}!}
\frac{c_{m-k+1}^{s_1}}{s_1!} \cdots \frac{c_{m}^{s_k}}{s_k!}
\mathbf{z}^{\mathbf{y}+\chi_K}.
\end{align*}
Replacing $\mathbf{y}$ by $\mathbf{y}+\chi_K$, we obtain \eqref{251209eq6}.
\end{proof}

\begin{lem}\label{lem:expansion-k}
Let $\CPP(x) = \sum_{j=0}^{l} a_j x^j$ be a univariate polynomial of degree $l$, 
and let the Hamiltonian $H$ be of the form in \eqref{eq:H-nearly-indep}.  Then
\begin{equation}\label{eq:PH-pauli-expansion-k}
\CPP(H) = \sum_{\mathbf{y} \in \mathbb{F}_2^{m-k}} w_{\mathbf{y}}\, P_{\mathbf{y}},
\end{equation}
where $P_{\mathbf{y}} := \prod_{i \in \mathrm{supp}(\mathbf{y})} P_i$ and the coefficient $w_{\mathbf y}$ is defined as
\begin{align}\label{eq:coeff-k}
w_{\mathbf{y}}=\sum_{t=0}^{l} a_t \cdot t! \sum_{K \subseteq \{1, \ldots, k\}} I^t_K(\mathbf{y})
\end{align}
with 
\begin{align}\label{eq:IK-def}
I^t_K(\mathbf{y}) 
=
\sum_{
\substack{
\mathbf{s} \in \Z_{\ge 0}^{k}, |\mathbf{s}|\leq t,\\
s_j\text{ is odd if } j\in K,\\
s_j\text{ is even if } j\not\in K.
}} 
\sum_{
\substack{
\boldsymbol{\mu} \in \Z_{\ge 0}^{m-k},\\
|\boldsymbol{\mu}|=t-|\mathbf{s}|,\\
\boldsymbol{\mu} \equiv \mathbf{y}+ \chi_{K}(\bmod 2) 
}} 
\prod ^{m-k}_{i=1}\frac{c_i^{\mu_i}}{\mu_i!}\prod^k_{i=1} \frac{c_{m-k+i}^{s_i}}{s_i!}.
\end{align} 
with $\boldsymbol{\mu} = (\mu_1,...,\mu_{m-k})\in \Z_{\ge 0}^{m-k}$, $|\boldsymbol{\mu} |:= \sum_{i=1}^{m-k} \mu_i$,  and
\begin{align}
\chi_{K} \equiv \sum_{j\in K} \;\mathbf{1}_{  U_j } (\bmod 2).
\end{align}
\end{lem}

\begin{proof}
This follows from Lemma~\ref{lem:expansion-k1} by substituting $z_i \to P_i$.
\end{proof}

Based on this lemma,
we now establish the MPS structure and preparation complexity of the reference 
state for the general nearly-independent commuting case. The construction extends 
Theorem~\ref{thm:ref-state} from Sec.~\ref{sec:com} by introducing a 
block-diagonal structure that tracks the parity pattern $K \subseteq \{1, \ldots, k\}$ 
of the redundant terms.

\begin{thm}[Reference state for nearly-independent commuting Hamiltonians]
\label{thm:ref-state-nearly-indep}
Let $H$ be a Hamiltonian of the form~\eqref{eq:H-nearly-indep} with symplectic code of dimension $k$, and let $\CPP(x)=\sum_{j=0}^{l}a_jx^j$ be a real polynomial of degree $l\ge1$ with $\CPP(H)\ne0$.
Consider the $(m-k)$-qubit state
\begin{equation}\label{eq:ref-state-k}
|R^l_k(H)\rangle 
= \frac{1}{\mathcal{N}} \sum_{\mathbf{y} \in \mathbb{F}_2^{m-k}} 
   \left( \sum_{t=0}^{l} a_t \cdot t! \sum_{K \subseteq \{1, \ldots, k\}} I^t_K(\mathbf{y}) \right) |\mathbf{y}\rangle,
\end{equation}
where $\mathcal{N}$ is a normalization factor and $I^t_K(\mathbf{y})$ is defined 
as  Eq.~\eqref{eq:IK-def}.
Then:
\begin{enumerate}[(i)]
\item \textbf{(MPS structure.)} $|R^l_k(H)\rangle$ is a 2-ary matrix product state 
with open boundary conditions and bond dimension $D = 2^k (l + 1)$. Specifically,
\begin{equation}\label{eq:mps-form-k}
|R^l_k(H)\rangle 
= \sum_{\mathbf{y} \in \mathbb{F}_2^{m-k}} 
   v_L^{(k)\top} \mathcal{A}^{(1)}(y_1) \cdots \mathcal{A}^{(m-k)}(y_{m-k})\, \mathcal{B}^{(*)}\, v_R^{(k)}\, |\mathbf{y}\rangle,
\end{equation}
where the boundary vectors $v_L^{(k)}, v_R^{(k)} \in \mathbb{R}^{2^k (l+1)}$ are 
formed by stacking $2^k$ copies of the boundary vectors of 
Theorem~\ref{thm:ref-state}:
\begin{equation*}
v_L^{(k)} = \tfrac{1}{\mathcal{N}}\, \bigl( e_1, \ldots, e_1 \bigr)^\top, 
v_R^{(k)} = \bigl( v_R, \ldots, v_R \bigr)^\top \quad (2^k \text{ copies each}),
\end{equation*}
and the site matrices $\{\mathcal{A}^{(i)}(y_i)\}_{i=1}^{m-k}$ are block-diagonal 
with $2^k$ blocks labeled by subsets $K \subseteq \{1, \ldots, k\}$:
\begin{equation}\label{eq:site-matrix-k}
\mathcal{A}^{(i)}(y_i) = \mathrm{diag}\!\left( A^{(i)}\bigl(y_i + \chi_K(i) \!\!\!\mod 2\bigr) : K \subseteq \{1, \ldots, k\} \right),
\end{equation}
with $A^{(i)}(\cdot)$ as in Eq.~\eqref{eq:site-matrix} and $\chi_K(i)$ the $i$-th 
component of ${\chi}_K$ defined as $\chi_{K} \equiv \sum_{j\in K} \;\mathbf{1}_{  U_j } (\bmod 2).$
The matrix 
$\mathcal{B}^{(*)} = \prod_{j=1}^{k} \mathcal{B}^{(m-k+j)}$ is the product of $k$ 
block-diagonal matrices encoding the contributions of the redundant terms; see
Eq.~\eqref{eq:Bj-block-def} in the proof below for the explicit definition.
\item \textbf{(Efficient preparation.)} $|R^l_k(H)\rangle$ can be prepared by a classical pre-processing step taking time $T_{\mathrm{pre}}=O(2^{2k}ml^3)$, followed by a quantum algorithm running in time $T_{\mathrm{prepare}}=O(m\cdot\mathrm{poly}(2^k,l))$. 
In particular, when $k = O(\log n)$, both runtimes scale polynomially in $n$ 
and $l$.
\end{enumerate}
\end{thm}

\begin{proof}
We first define the matrices $\mathcal{B}^{(m-k+j)}$ explicitly. For each 
$j \in \{1, \ldots, k\}$, let
\begin{equation}\label{eq:Bj-block-def}
\mathcal{B}^{(m-k+j)} = \mathrm{diag}\!\left( B^{(m-k+j)}_K : K \subseteq \{1, \ldots, k\} \right),
\end{equation}
where each block $B^{(m-k+j)}_K \in \mathbb{R}^{(l+1) \times (l+1)}$ is given by
\begin{equation*}
B^{(m-k+j)}_K = 
\begin{cases}
A^{(m-k+j)}(0) & \text{if } j \notin K, \\
A^{(m-k+j)}(1) & \text{if } j \in K,
\end{cases}
\end{equation*}
with $A^{(\cdot)}(\cdot)$ as in Eq.~\eqref{eq:site-matrix}. The product 
$\mathcal{B}^{(*)} := \prod_{j=1}^{k} \mathcal{B}^{(m-k+j)}$ is block-diagonal 
with $2^k$ blocks: the block corresponding to $K \subseteq \{1, \ldots, k\}$ is
\begin{equation*}
[\mathcal{B}^{(*)}]_K = \prod_{j=1}^{k} B^{(m-k+j)}_K = \prod_{j=1}^{k} A^{(m-k+j)}\bigl(\mathbf{1}[j \in K]\bigr),
\end{equation*}
which encodes the $K$-pattern of parities for the $k$ redundant terms.

\noindent\textbf{MPS structure.}
We show that the coefficient of $|\mathbf{y}\rangle$ in Eq.~\eqref{eq:mps-form-k} 
equals that in Eq.~\eqref{eq:ref-state-k}. By the block-diagonal structure 
of $\mathcal{A}^{(i)}(y_i)$ and $\mathcal{B}^{(*)}$, the product 
$$\mathcal{A}^{(1)}(y_1) \cdots \mathcal{A}^{(m-k)}(y_{m-k})\, \mathcal{B}^{(*)}$$
is also block-diagonal, with $K$-th block equal to
\begin{align*}
&\bigl[ \mathcal{A}^{(1)}(y_1) \cdots \mathcal{A}^{(m-k)}(y_{m-k}) \mathcal{B}^{(*)} \bigr]_K \\
=& A^{(1)}\bigl(y_1 + \chi_K(1)\bigr) \cdots A^{(m-k)}\bigl(y_{m-k} + \chi_K(m-k)\bigr) \prod_{j=1}^{k} A^{(m-k+j)}\bigl(\mathbf{1}[j \in K]\bigr).
\end{align*}
Contracting with the boundary vectors $v_L^{(k)}, v_R^{(k)}$, which are block-stacked, 
gives
\begin{align*}
&v_L^{(k)\top} \mathcal{A}^{(1)}(y_1) \cdots \mathcal{A}^{(m-k)}(y_{m-k})\, \mathcal{B}^{(*)}\, v_R^{(k)}\\
= &\frac{1}{\mathcal{N}} \sum_{K \subseteq \{1, \ldots, k\}} 
   e_1^\top A^{(1)}(y_1 + \chi_K(1)) \cdots A^{(m-k)}(y_{m-k} + \chi_K(m-k))\prod_{j=1}^{k} A^{(m-k+j)}\bigl(\mathbf{1}[j \in K]\bigr)\, v_R.
\end{align*}
Applying the shift-matrix identity in Lemma~\ref{lem:shift-mult} to expand each 
$A^{(i)}(\cdot)$ as a polynomial in $B_j$ and using $e_1^\top B_j\, v_R = a_j \cdot j!$ 
gives
\begin{align*}
v_L^{(k)\top} \mathcal{A}^{(1)}(y_1) \cdots \mathcal{A}^{(m-k)}(y_{m-k})\, \mathcal{B}^{(*)}\, v_R^{(k)}
= \frac{1}{\mathcal{N}} \sum_{t=0}^{l} a_t \cdot t! \sum_{K \subseteq \{1, \ldots, k\}} I^t_K(\mathbf{y}),
\end{align*}
which matches the coefficient of $|\mathbf{y}\rangle$ in Eq.~\eqref{eq:ref-state-k}.

The bond dimension is $D = 2^k (l + 1)$ since each $\mathcal{A}^{(i)}(y_i)$ is a 
$2^k$-fold block diagonal matrix with $(l+1) \times (l+1)$ blocks.

\noindent\textbf{Preparation complexity.}
We analyze classical pre-processing and quantum preparation separately.

\medskip
\noindent\textbf{Classical pre-processing.}

\emph{(a) Site matrices.} Each $\mathcal A^{(i)}(y_i)$ has $2^k$ blocks of size $(l+1)\times(l+1)$, and the entries are determined by the coefficients $c_i^j/j!$.
Computing the coefficients and writing all site matrices takes $O(2^km(l+1)^2)$ operations.

\emph{(b) Matrix $\mathcal B^{(*)}$.} The product $\mathcal B^{(*)}=\prod_{j=1}^k\mathcal B^{(m-k+j)}$ can be computed in $O(k2^kl^3)$ operations using its block-diagonal structure.

\emph{(c) Normalization.} The normalization factor satisfies
\begin{align*}
\mathcal{N}^2 &= v_R^{(k)\top}\, \mathcal{B}^{(*)\top}\! \Bigg( \sum_{y_{m-k}} \mathcal{A}^{(m-k)}(y_{m-k})^\top \cdots  \left( \sum_{y_1} \mathcal{A}^{(1)}(y_1)^\top (e_1,...,e_1) (e_1,...,e_1)^\top \mathcal{A}^{(1)}(y_1) \right)\cdots\\
& \times \mathcal{A}^{(m-k)}(y_{m-k}) \Bigg) \mathcal{B}^{(*)}\, v_R^{(k)},
\end{align*}
This expression is evaluated from inside out.
At each step, the intermediate matrix has $2^k\times2^k$ blocks of size $(l+1)\times(l+1)$, so the cost is $O(2^{2k}l^3)$ per site and $O(2^{2k}ml^3)$ in total.

\emph{(d) Circuit construction.} Set $Q=2^k$, $d=l+1$, and $s=m-k$, and let $h_i=(\mathbf{1}_{i\in U_j})_{j=1}^k\in\mathbb F_2^k$.
Let $F=H_{\mathrm{Had}}^{\otimes k}$, where $H_{\mathrm{Had}}$ is the one-qubit Hadamard matrix, and absorb $\mathcal B^{(*)}$ into the right boundary vector.
Set $C_i=(A^{(i)}(0)+A^{(i)}(1))/\sqrt2$ and $S_i=(A^{(i)}(0)-A^{(i)}(1))/\sqrt2$, and define $X_{h_i}|u\rangle=|u+h_i\rangle$ for $u\in\mathbb F_2^k$.
Applying $F$ to the bond indices and a Hadamard change of basis to the physical index gives the following matrices.
\[ \widehat{\mathcal A}^{(i)}(0)=I_Q\otimes C_i,\qquad \widehat{\mathcal A}^{(i)}(1)=X_{h_i}\otimes S_i. \]
The unnormalized left boundary becomes $\sqrt Q\,|0\rangle\otimes e_1$.
Let $G_0=Q|0\rangle\langle0|\otimes e_1e_1^\top$ and $G_i=\sum_{b=0}^1\widehat{\mathcal A}^{(i)}(b)^\top G_{i-1}\widehat{\mathcal A}^{(i)}(b)$.
Each $G_i$ is block diagonal, and its $u$-th block satisfies the following recurrence.
\[ G_{i,u}=C_i^\top G_{i-1,u}C_i+S_i^\top G_{i-1,u+h_i}S_i. \]
For each nonzero block, choose a full-row-rank factor $R_{i,u}$ such that $G_{i,u}=R_{i,u}^\top R_{i,u}$, and let $R_{i,u}^+$ be its right inverse; zero-rank blocks are omitted.
For each $u$ with $\operatorname{rank}G_{i,u}>0$, define the following matrices, where a missing zero-rank output block is interpreted as zero.
\[ M_{i,u}=R_{i-1,u}C_iR_{i,u}^+,\qquad N_{i,u}=R_{i-1,u+h_i}S_iR_{i,u}^+. \]
These matrices satisfy $M_{i,u}^\top M_{i,u}+N_{i,u}^\top N_{i,u}=I$.
Thus $M_{i,u}$ and $N_{i,u}$ define a site isometry whose two outputs have bond labels $u$ and $u+h_i$.
After padding and conditioned on $u$, each such isometry can be extended to a unitary on a space of dimension at most $4d$, with a controlled $X_{h_i}$ supplying the change of bond label.
The initial bond state is obtained by applying the factors $R_{s,u}$ to the corresponding blocks of $(F\otimes I_d)\mathcal B^{(*)}v_R^{(k)}$ and normalizing.
Applying the site isometries in reverse order and undoing the physical Hadamard changes prepares $|R_k^l(H)\rangle$ and leaves a one-dimensional bond space.
The factors and unitary extensions can be computed in $O(mQd^3)$ arithmetic operations.
Using $O(d^2)$ two-level gates for each extension and standard controlled-gate decompositions with $O(k+1+\log_2 d)$ work qubits, the classical cost of constructing the circuit, including its initial state preparation, is $O(mQd^3+mQd^2(k+1+\log_2 d)^2)$.
Since $(k+1+\log_2 d)^2=O(Qd)$, the circuit-construction cost is $O(mQ^2d^3)=O(2^{2k}ml^3)$.
Combining (a)--(d) gives $T_{\mathrm{pre}}=O(2^{2k}ml^3)$.

\medskip
\noindent\textbf{Quantum preparation.} A 2-ary MPS with $m$ sites, open boundary 
conditions, and bond dimension $D$ can be prepared by a quantum circuit running 
in time $O(m \cdot \mathrm{poly}(D))$~\cite{schon2005sequential, melnikov2023quantum}. 
Setting $D = 2^k (l + 1)$ gives $T_{\mathrm{prepare}} = O(m \cdot \mathrm{poly}(2^k, l))$.

When $k = O(\log n)$, both $2^k = \mathrm{poly}(n)$ and the runtimes are polynomial 
in $n$ and $l$.
\end{proof}


\begin{Rem}
The factor $2^k$ in the bond dimension is the price paid for encoding the $k$ 
linear relations among the Pauli operators. When $k = 0$ (the fully independent 
case of Theorem~\ref{thm:hdqi-comp_indep}), the construction reduces exactly to 
Theorem~\ref{thm:ref-state}, with bond dimension $D = l + 1$ and no parity tracker.
\end{Rem}

\begin{thm}\label{thm:hdqi-nearly-indep}
Let $H$ be a Hamiltonian of the form~\eqref{eq:H-nearly-indep} with symplectic code of dimension $k$, and let $\CPP(x)=\sum_{j=0}^{l}a_jx^j$ be a real polynomial of degree $l\ge1$ with $\CPP(H)\ne0$.
Then there exists a quantum algorithm that prepares the state
\begin{equation*}
\rho_{\CPP}(H) = \frac{\CPP^2(H)}{\mathrm{Tr}[\CPP^2(H)]}.
\end{equation*}
The algorithm uses one call to the weight-$(m-k)$ decoding oracle $\mathcal D_{H'}$ for the independent part $H'=\sum_{i=1}^{m-k}c_iP_i$.
The reference-state pre-processing time is $O(2^{2k}ml^3)$, and the quantum runtime apart from the oracle call is $O(m\cdot\mathrm{poly}(2^k,l,n))$.
The oracle can be implemented by Gaussian elimination, so the algorithm runs in polynomial time when $k=O(\log n)$ and $l=\mathrm{poly}(n)$.
\end{thm}

\begin{proof}
The construction follows that of Theorem~\ref{thm:hdqi-commuting}, with the 
reference state $|R^l(H)\rangle$ replaced by $|R^l_k(H)\rangle$ from 
Theorem~\ref{thm:ref-state-nearly-indep}. The key observations are:
\begin{itemize}
\item By Lemma~\ref{lem:expansion-k}, the operator $\CPP(H)$ admits the Pauli 
expansion
\begin{equation*}
\CPP(H) = \sum_{\mathbf{y} \in \mathbb{F}_2^{m-k}} 
   w_{\mathbf y}P_{\mathbf{y}},
\end{equation*}
where $ w_{\mathbf y}$ is defined in Eq.~\eqref{eq:coeff-k}, and  $P_{\mathbf{y}} = \prod_{i \in \mathrm{supp}(\mathbf{y})} P_i$ involves 
only the $m-k$ independent Pauli operators.
\item Since these $m-k$ operators are linearly independent in $\mathbb{F}_2^{2n}$, 
the weight-$(m-k)$ decoding oracle $\mathcal{D}_{H'}$ exists and recovers any 
$\mathbf{y} \in \mathbb{F}_2^{m-k}$ from $\mathrm{symp}(P_{\mathbf{y}})$ (cf.\ 
Theorem~\ref{thm:hdqi-comp_indep}).
\end{itemize}
The state evolution in Steps 2--7 of the proof of Theorem~\ref{thm:hdqi-commuting} proceeds with these substitutions, producing $\rho_\CPP(H)$ on register $B$ after register $C$ is traced out.
The complexity includes the reference state preparation in Theorem~\ref{thm:ref-state-nearly-indep}(ii), the controlled-Pauli operations, the Bell transform, and a decoder implemented by Gaussian elimination.
\end{proof}

\section{Noncommuting Hamiltonians}\label{sec:noncom}
In this section, we consider the Hamiltonian $H=\sum_ic_iP_i$ with real coefficients $c_i\in \real$ where the Pauli operators $P_i$
 are not necessarily mutually commuting. The presence of noncommuting Pauli terms introduces additional challenges, particularly in computing the coefficients appearing in the symmetrized polynomial expansion.
Our key observation is that noncommuting Pauli terms can be organized into anticommuting clusters, each of which contributes locally to the reference state. The tensor-network structure arises from composing these clusters sequentially.

\begin{Def}[Anticommutation graph, \cite{schmidhuber2025hamiltonian}]\label{def:anti_graph}
Let $z_1,\ldots,z_m$ be formal variables such that every pair of distinct variables either commutes or anticommutes, i.e.,
$$z_iz_j=\pm z_jz_i, \quad i\neq j.$$
We define the anticommutation graph of $\{z_1,...,z_m\}$ as an undirected graph $G=(V,E)$ with vertex set $V = \{z_1,...,z_m\}$ and two vertices are connected by an edge $(z_i,z_j)\in E$ if and
only if $ z_i z_j = - z_j z_i$.

\end{Def}

Furthermore, 
for any formal variables $z_1,\ldots,z_m$ with $z_i^2=1$ for all $i$,
let $\bmu=(\mu_1,\cdots, \mu_m)\in\Z_{\geq 0}^m$ be a counting vector,
then we denote $S(\bmu)$ to be the set of all strings over the index set $\{1,...,m\}$ which contains $i$ for $\mu_i$ times.
For any $\pi\in S(\bmu)$, the sign function associated with $G$ and $\bmu$ is defined by
\begin{align}
\sgn_{G,\bmu}(\pi) := (z_m^{\mu_m} z_{m-1}^{\mu_{m-1}} \cdots z_1^{\mu_1}) (z_{\pi(1)} \cdots z_{\pi(|{\bmu}|)}) \in \set{\pm 1} .
\end{align}

With the above notation, we are now ready to derive the coefficient formulas.
\begin{Def}[Weighted $\beta$-function] \label{def:weighted-beta}
Let $z_1, \dots, z_m$ be a set of formal variables with anti-commutation graph $G$, and let $\mathbf{c} = (c_1,...,c_m)\in \R^m$. 
For $s \ge 0$, define the weighted $\beta$-function of order $s$ as 
\begin{align}
\beta_G^{(s)}(\mathbf{y},\mathbf{c}) 
= \sum_{
\substack{\boldsymbol{\mu} \in \Z_{\ge 0}^{m },\\
|\boldsymbol{\mu}|=s \\
\bmu   \equiv \mathbf{y} (\bmod 2)}}
\mathbf{c}^{\bmu}
 \sum_{\pi\in S(\bmu)}
\sgn _{G,\bmu}(\pi), \quad \forall \mathbf{y} \in \BFF_2^m.
\end{align}

\end{Def}

Using the anticommutation graph and weighted $\beta$-functions, we can rewrite $\CPP(H)$ for noncommuting Hamiltonians by the following lemma, which is a noncommutative generalization of Lemma
\ref{lem:poly-expansion}.

\begin{lem}[Noncommuting univariate polynomial expansion] 
\label{lem:noncom_poly_expan}
Let $\CPP(x) = \sum^l_{j=0} a_jx^j$ be a univariate polynomial of degree $l$.
Let $x = \sum_{i=1}^m c_i z_i$ such that $z_i^2 = 1$ for all $i$ and  $\set{z_i}_i$
satisfy the (anti) commutation relations encoded by an anticommutation graph $G$, as defined in Definition \ref{def:anti_graph}. 
Then the polynomial  $\CPP(x)$ admits the expansion
\begin{align}
\CPP\left( \sum_{i=1}^m c_i z_i \right) & = \sum_{s=0}^l a_s \, \,\sum_{\mathbf{y} \in \BFF_2^m}
    \beta_G^{(s)}(\mathbf{y}, \mathbf{c}) \cdot
    z_{1}^{y_1} \cdots z_{m}^{y_m} ,
\end{align}
where $\beta_G^{(s)}(\mathbf{y},\mathbf{c})$ denotes the weighted $\beta$-function of order $s$, as defined in Definition~\ref{def:weighted-beta}.
\end{lem}

\begin{proof}
We have
\begin{align}
\label{260102eq1}\CPP\left( \sum_{i=1}^m c_i z_i \right)  =\sum_{s=0}^l a_s \left( \sum_{i=1}^m c_i z_i \right)^s 
=\sum_{s=0}^l a_s \sum_{1 \leq t_1, \dots, t_s \leq m} c_{t_1} \cdots c_{t_s} z_{t_1} \cdots z_{t_s}.
\end{align}

Since
\begin{align*}
\sum_{1 \leq t_1, \dots, t_s \leq m} 
c_{t_1} \cdots c_{t_s}
z_{t_1} \cdots z_{t_s}
=& \sum_{\substack{
\boldsymbol{\mu} \in \Z_{\ge 0}^{m },
|\boldsymbol{\mu}|=s}}
\sum_{\pi \in S(\bmu)} 
\sgn_{G, \bmu}(\pi) 
c_1^{\mu_1} \cdots c_m^{\mu_m}
z_1^{\mu_1} \cdots z_m^{\mu_m} \\
=& \sum_{\mathbf{y} \in \BFF_2^m} \beta_{G}^{(s)}(\mathbf{y},\mathbf{c}) \cdot \mathbf{z}^{\mathbf{y}}.
\end{align*}
Together with equation \eqref{260102eq1}, the proof is completed.
\end{proof}

\begin{lem}[Operator polynomial expansion, noncommuting case]
\label{lem:operator-noncomm-expansion}
Let $\CPP(x) = \sum_{j=0}^{l} a_j x^j$ be a univariate polynomial of degree $l$, 
and let $H = \sum_{i=1}^{m} c_i P_i$ be a Hamiltonian on $n$ qubits with 
$c_i \in \mathbb{R}$ and distinct Pauli operators $\{P_i\}$. Let $G$ denote 
the anticommutation graph of $H$. For each $\mathbf{y} \in \mathbb{F}_2^m$, 
define the ordered Pauli product
\begin{equation}\label{eq:Py-ordered}
P_{\mathbf{y}} := P_1^{y_1} P_2^{y_2} \cdots P_m^{y_m},
\end{equation}
where the factors are taken in increasing index order. Then
\begin{equation}\label{eq:PH-noncomm-expansion}
\CPP(H) = \sum_{s=0}^{l} a_s \sum_{\mathbf{y} \in \mathbb{F}_2^m} \beta_G^{(s)}(\mathbf{y}, \mathbf{c})\, P_{\mathbf{y}},
\end{equation}
where $\beta_G^{(s)}(\mathbf{y}, \mathbf{c})$ is the weighted $\beta$-function 
of Definition~\ref{def:weighted-beta}. In particular, for any $s > l$, the 
contribution of order $s$ vanishes, so only terms with $|\mathbf{y}| \le l$ 
contribute to the sum.
\end{lem}

\begin{proof}
This follows from Lemma~\ref{lem:noncom_poly_expan} by substituting $z_i \to P_i$.

\end{proof}

\begin{lem}[Computational cost of the weighted $\beta$-function] \label{lem:beta-cost}
Let $\mathbf{y} \in \BFF_2^M$ such that $|\mathbf{y}| = w$, and let $ \beta_{G}^{(s)}(\mathbf{y}, \mathbf{c})$ be as in Definition~\ref{def:weighted-beta}.
Assume that elementary arithmetic operations and classical memory access take $O(1)$ time.
There is a deterministic classical algorithm which computes $ \beta_{G}^{(s)}(\mathbf{y}, \mathbf{c})$ in time $T$, where
\begin{align}\label{260104eq4}
T = \begin{cases}
O\left( \binom{\frac{s-w}{2} + M - 1}{M - 1} \cdot s M 2^s \right) & \text{ if } s \leq M, \\
O\left( \binom{\frac{s-w}{2} + M - 1}{M - 1} \cdot s M (\lceil s/M \rceil + 1)^M \right) & \text{ if } s > M .
\end{cases}
\end{align}
\end{lem}
\begin{proof}
From Definition \ref{def:weighted-beta},
\begin{align}\label{260104eq5}
\beta_{G}^{(s)}(\mathbf{y}, \mathbf{c}) 
=  \sum_{
\substack{\boldsymbol{\mu} \in \Z_{\ge 0}^{M },\\
|\boldsymbol{\mu}|=s \\
\bmu   \equiv \mathbf{y} (\bmod 2)}}
\binom{s}{\bmu} \cdot c_1^{\mu_1}\cdots c_M^{\mu_M}\alpha_G(\bmu ),
\end{align}
where 
\begin{align}
    \alpha_G(\bmu) = \frac{1}{|S(\bmu)|}\sum_{\pi\in S(\bmu)} \sgn_{G, \bmu}(\pi).
\end{align}
First,
each $\binom{s}{\bmu}$ can be computed in $O(m \cdot \mathrm{poly}(s))$ time.
Now let us count the number of terms in the summation in \eqref{260104eq5}.
Let $w = |\mathbf{y}|$ be the weight of $\mathbf{y}$.
If $w > s$, then the summation has no terms.
If $w \le s$, let $\boldsymbol{\nu} \in \Z_{\geq0}^{M}$ such that $\bmu = \mathbf{y} + 2\boldsymbol{\nu}$,
then $|\boldsymbol{\nu}| = \frac{s-w}{2}$.
If $\frac{s-w}{2} \notin \Z_{\geq0}$, then there are no such $\boldsymbol{\nu}$ or $\boldsymbol{\mu}$.
Otherwise, there are $\binom{\frac{s-w}{2} + M - 1}{M - 1}$ such $\boldsymbol{\nu}$,
so there are $\binom{\frac{s-w}{2} + M - 1}{M - 1}$ terms in the summation in \eqref{260104eq5}.
Combined with Theorem 25 in \cite{schmidhuber2025hamiltonian}, the total runtime of calculating the weighted $\beta$-function is bounded by \eqref{260104eq4}.
\end{proof}

\begin{lem}[Factorization of weighted $\beta$-function] \label{lem:beta-factorization}
Let $z_1, \dots, z_m$ be a set of $m$ formal variables with anticommutation graph $G = (V, E)$.
Let  $G_1, \dots, G_r$ denote the connected components of $G$, where 
each component is given by $G_t = (V_t, E_t)$,  and the vertex and edge sets decompose as $V = \bigsqcup_{t=1}^r V_t$, $E = \bigsqcup_{t=1}^r E_t$.
Let $m_t = |V_t|$ denote the number of vertices in $G_t$, so that $\sum_{t=1}^r m_t = m$.
Then, for any $\mathbf{y} \in \BFF_2^m$, the weighted $\beta$-function factorizes as 
\begin{align}
\beta_{G}^{(s)}(\mathbf{y},\mathbf{c}) = 
\sum_{
\substack{
\bkappa \in \Z_{\ge 0}^r, \\|\bkappa| =s}} \binom{s}{\bkappa} \prod_{t=1}^r \beta_{G_t}^{(\kappa_t)}(\mathbf{y}\rvert_{G_t},
\mathbf{c}\rvert_{G_t}) ,
\end{align}
where $\mathbf{y}\rvert_{G_t}$ (respectively, $\mathbf{c}\rvert_{G_t}$) denotes the restriction of $\mathbf{y}$ (respectively, $\mathbf{c} $) to the indices  corresponding to the vertices in $G_t$.

\end{lem}
\begin{proof}
Denote $\kappa_t = \left|\bmu\rvert_{G_t}\right|$ to be the size of the restriction of $\bmu$ onto $G_t$.
Therefore $\sum_{t=1}^r \kappa_t = |\bmu|= s$.
Since 
\begin{align*}
\prod_{t=1}^r \beta_{G_t}^{(\kappa_t)} (\mathbf{y}\rvert_{G_t},
\mathbf{c}\rvert_{G_t}) 
=& \prod_{t=1}^r \sum_{\substack{
\bmu \in \Z_{\ge 0}^{m_t}, \, |\bmu|=\kappa_t \\
\bmu \equiv \mathbf{y}\rvert_{G_t} (\bmod 2)}} 
(\mathbf{c}\rvert_{G_t})^{\bmu}\sum_{\pi\in S(\bmu)} \sgn_{G_t, \bmu} (\pi)\\
=&  \sum_{\substack{\bmu\in\Z_{\geq 0}^{m}\\
|\bmu\rvert_{G_j}|=\kappa_j,\text{ for }j=1,\cdots, r\\
\bmu \equiv \mathbf{y} (\bmod 2)}}
\mathbf{c}^{\bmu} 
\left( \frac{1}{\binom{s}{\bkappa}}
\sum_{\pi\in S(\bmu)} \sgn_{G, \bmu}(\pi)
\right).
\end{align*}
Then, we have that 
\begin{align*}
\sum_{
\substack{
\bkappa \in \Z_{\ge 0}^r, \\|\bkappa| =s}} \binom{s}{\bkappa} \prod_{t=1}^r \beta_{G_t}^{(\kappa_t)}(\mathbf{y}\rvert_{G_t},
\mathbf{c}\rvert_{G_t}) 
=&  \sum_{
\substack{
\bkappa \in \Z_{\ge 0}^r, \\|\bkappa| =s}} \binom{s}{\bkappa} 
\sum_{\substack{\bmu\in\Z_{\geq 0}^{m}\\
|\bmu\rvert_{G_j}|=\kappa_j,\text{ for }j=1,\cdots, r\\
\bmu \equiv \mathbf{y} (\bmod 2)}}
\mathbf{c}^{\bmu} 
\left( \frac{1}{\binom{s}{\bkappa}}
\sum_{\pi\in S(\bmu)} \sgn_{G, \bmu}(\pi)
\right)\\
= & \sum_{\substack{\bmu\in\Z_{\geq 0}^m,\; |\bmu|=s\\
\bmu\equiv\mathbf{y}(\bmod 2)}}
\mathbf{c}^{\bmu} \sum_{\pi\in S(\bmu)} \sgn_{G,\bmu}(\pi)
=  \beta_{G}^{(s)}(\mathbf{y},\mathbf{c}).
\end{align*}
\end{proof}

We now establish the MPS structure and preparation complexity of the reference 
state for general (possibly noncommuting) Pauli Hamiltonians. The construction 
extends Theorem~\ref{thm:ref-state} by replacing the parity-based shift matrices 
with site matrices that encode the weighted $\beta$-function, with each "site" 
of the MPS corresponding to a connected component of the anticommutation graph.

\begin{thm}[Reference state for noncommuting Hamiltonians]
\label{thm:ref-state-noncomm}
Let $\CPP(x)=\sum_{j=0}^{l}a_jx^j$ be a real polynomial of degree $l$, and let $H=\sum_{i=1}^{m}c_iP_i$ be an $n$-qubit Hamiltonian with $\CPP(H)\ne0$ and anticommutation graph $G=(V,E)$.
Let $G=\bigsqcup_{t=1}^rG_t$ be the decomposition into connected components $G_t=(V_t,E_t)$, and set $\CMM:=\max_{1\le t\le r}|V_t|$. 
Consider the state
\begin{equation}\label{eq:ref-state-star}
|R^l_*(H)\rangle 
= \frac{1}{\mathcal{N}} \sum_{s=0}^{l} a_s \sum_{\mathbf{y} \in \mathbb{F}_2^m} \beta_G^{(s)}(\mathbf{y}, \mathbf{c})\, |\mathbf{y}\rangle,
\end{equation}
where $\mathcal{N}$ is a normalization factor and $\beta_G^{(s)}(\mathbf{y}, \mathbf{c})$ 
is the weighted $\beta$-function of Definition~\ref{def:weighted-beta}. Then:
\begin{enumerate}[(i)]
\item \textbf{(MPS structure.)} $|R^l_*(H)\rangle$ is a $q$-ary matrix product state 
with open boundary conditions, bond dimension $D = l + 1$, and local dimension 
$q = 2^{\CMM}$. Specifically,
\begin{equation}\label{eq:mps-form-star}
|R^l_*(H)\rangle 
= \sum_{\mathbf{y} \in \mathbb{F}_2^m} 
   v_L^\top A^{(1)}(\mathbf{y}^{(1)}) \cdots A^{(r)}(\mathbf{y}^{(r)})\, v_R\, |\mathbf{y}\rangle,
\end{equation}
where $\mathbf{y}^{(t)} := \mathbf{y}|_{G_t}$ denotes the restriction of $\mathbf{y}$ 
to the indices of component $G_t$, the boundary vectors are
\begin{equation*}
v_L = \tfrac{1}{\mathcal{N}}\, e_1, v_R = (a_0, a_1, \ldots, a_l)^\top, 
 e_1 = (1, 0, \ldots, 0)^\top \in \mathbb{R}^{l+1},
\end{equation*}
and the site matrices $A^{(t)}(\mathbf{y}^{(t)}) \in \mathbb{R}^{(l+1) \times (l+1)}$ 
are upper-triangular with entries
\begin{equation}\label{eq:site-matrix-star}
\bigl[ A^{(t)}(\mathbf{y}^{(t)}) \bigr]_{i, j} 
= \begin{cases}
\binom{j}{i} \cdot \beta_{G_t}^{(j-i)}(\mathbf{y}^{(t)}, \mathbf{c}|_{G_t}) & \text{if } j \ge i, \\
0 & \text{if } j < i,
\end{cases}
\end{equation}
Here $ 0 \le i, j \le l$, and $\beta_{G_t}^{(j-i)}(\mathbf{y}^{(t)}, \mathbf{c}|_{G_t}) $ is given in Definition \ref{def:weighted-beta} with $\mathbf{c}\rvert_{G_t}$ denoting the restriction of $\mathbf{c} $ to the indices  corresponding to the vertices in $G_t$.
\item \textbf{(Efficient preparation.)} $|R^l_*(H)\rangle$ can be prepared by 
a classical pre-processing step taking time $T_{\mathrm{pre}} = O(m) \cdot (l + \mathcal M)^{O(\mathcal M)}$, 
followed by a quantum algorithm running in time $T_{\mathrm{prepare}} = O(m \cdot 2^{\mathcal M} \cdot \mathrm{poly}(l))$. 
In particular, when $\mathcal M = O(1)$ and $l, m = \mathrm{poly}(n)$, both runtimes scale 
polynomially in $n$.
\end{enumerate}
\end{thm}

\begin{proof}
We first prove part (i), concerning the MPS structure.
We show that the coefficient of $|\mathbf{y}\rangle$ in Eq.~\eqref{eq:mps-form-star} 
matches that in Eq.~\eqref{eq:ref-state-star}. Expanding the matrix product 
using the entries~\eqref{eq:site-matrix-star},
\begin{align*}
v_L^\top A^{(1)}(\mathbf{y}^{(1)}) \cdots A^{(r)}(\mathbf{y}^{(r)})\, v_R 
=&  \frac{1}{\mathcal{N}} \sum_{s=0}^{l} \sum_{0 \le K_1 \le K_2 \le \cdots \le K_{r-1} \le s \le l} 
   \binom{K_1}{0} \beta_{G_1}^{(K_1)}(\mathbf{y}^{(1)}, \mathbf{c}|_{V_1}) \\
&~\times \binom{K_2}{K_1} \beta_{G_2}^{(K_2 - K_1)}(\mathbf{y}^{(2)}, \mathbf{c}|_{V_2}) \cdots \binom{s}{K_{r-1}} \beta_{G_r}^{(s - K_{r-1})}(\mathbf{y}^{(r)}, \mathbf{c}|_{V_r}) \cdot a_s.
\end{align*}
Substituting $\kappa_t := K_t - K_{t-1}$ (with $K_0 := 0$ and $K_r := s$), so 
that $\kappa_1 + \cdots + \kappa_r = s$, and using the identity
\begin{equation*}
\binom{K_1}{0} \binom{K_2}{K_1} \cdots \binom{s}{K_{r-1}} 
= \binom{s}{\kappa_1, \kappa_2, \ldots, \kappa_r} 
= \frac{s!}{\kappa_1! \cdots \kappa_r!},
\end{equation*}
we obtain
\begin{align*}
v_L^\top A^{(1)}(\mathbf{y}^{(1)}) \cdots A^{(r)}(\mathbf{y}^{(r)})\, v_R 
=& \frac{1}{\mathcal{N}} \sum_{s=0}^{l} a_s 
   \sum_{\substack{\boldsymbol{\kappa} \in \mathbb{Z}_{\ge 0}^r \\ |\boldsymbol{\kappa}| = s}} 
   \binom{s}{\boldsymbol{\kappa}} \prod_{t=1}^{r} \beta_{G_t}^{(\kappa_t)}(\mathbf{y}^{(t)}, \mathbf{c}|_{V_t}).
\end{align*}
Applying Lemma~\ref{lem:beta-factorization} (factorization of the weighted 
$\beta$-function), the inner sum equals $\beta_G^{(s)}(\mathbf{y}, \mathbf{c})$. 
Therefore,
\begin{equation*}
v_L^\top A^{(1)}(\mathbf{y}^{(1)}) \cdots A^{(r)}(\mathbf{y}^{(r)})\, v_R 
= \frac{1}{\mathcal{N}} \sum_{s=0}^{l} a_s\, \beta_G^{(s)}(\mathbf{y}, \mathbf{c}),
\end{equation*}
which matches the coefficient of $|\mathbf{y}\rangle$ in Eq.~\eqref{eq:ref-state-star}.

The bond dimension is $D = l + 1$ since each $A^{(t)}(\mathbf{y}^{(t)}) \in \mathbb{R}^{(l+1) \times (l+1)}$. 
The local dimension is $q = \max_t 2^{|V_t|} = 2^{\mathcal M}$ since the MPS site associated 
with component $G_t$ carries the variable $\mathbf{y}^{(t)} \in \mathbb{F}_2^{|V_t|}$.

We now turn to part (ii), concerning the preparation complexity.
We analyze classical pre-processing and quantum preparation separately.

\medskip
\noindent\textbf{Classical pre-processing.}

\emph{(a) Site matrix entries.} Each nonzero entry of $A^{(t)}(\mathbf{y}^{(t)})$ 
has the form $\binom{j}{i} \cdot \beta_{G_t}^{(j-i)}(\mathbf{y}^{(t)}, \mathbf{c}|_{V_t})$ 
for $0 \le i \le j \le l$. The binomial coefficient $\binom{j}{i}$ is computable 
in $O(l)$ time. By Lemma~\ref{lem:beta-cost}, each $\beta_{G_t}^{(j-i)}$ is 
computable in time at most
\begin{equation*}
O\!\left( \binom{(l-w)/2 + \mathcal M - 1}{\mathcal M - 1} \cdot l\, \mathcal M\, l^{\mathcal M} \right) 
= (l + \mathcal M)^{O(\mathcal M)},
\end{equation*}
where $w = |\mathbf{y}^{(t)}|$. Hence each entry costs at most $(l + \mathcal M)^{O(\mathcal M)}$, 
and computing a full site matrix $A^{(t)}(\mathbf{y}^{(t)})$ (containing $O(l^2)$ 
nonzero entries) takes $O(l^2) \cdot (l + \mathcal M)^{O(\mathcal M)} = (l + \mathcal M)^{O(\mathcal M)}$. With 
$r \le m$ such matrices and $O(2^{\mathcal M})$ choices of $\mathbf{y}^{(t)}$ per matrix, 
the total cost is
\begin{equation*}
O(m\, 2^{\mathcal M}) \cdot (l + \mathcal M)^{O(\mathcal M)} = O(m) \cdot (l + \mathcal M)^{O(\mathcal M)}.
\end{equation*}

\emph{(b) Boundary vectors.} The vector $v_R = (a_0, a_1, \ldots, a_l)^\top$ 
requires $O(l)$ time. The normalization factor is
\begin{align*}
\mathcal{N}^2 =& v_R^\top \!\left( \sum_{\mathbf{y}^{(r)}} A^{(r)}(\mathbf{y}^{(r)})^\top \cdots \left( \sum_{\mathbf{y}^{(1)}} A^{(1)}(\mathbf{y}^{(1)})^\top e_1 e_1^\top A^{(1)}(\mathbf{y}^{(1)}) \right)\right.  \left.\cdots A^{(r)}(\mathbf{y}^{(r)}) \right) v_R,
\end{align*}
evaluated from inside out. Each summation step computes a sum over $\mathbf{y}^{(t)} \in \mathbb{F}_2^{|V_t|}$ 
of $O(2^{\mathcal M})$ terms, each involving two $(l+1) \times (l+1)$ matrix multiplications 
at cost $O(l^3)$. With $r \le m$ such steps, the total cost is $O(2^{ \mathcal  M}\, m\, l^3)$.

Combining (a) and (b), the total classical pre-processing time is
\begin{equation*}
T_{\mathrm{pre}} = O(m) \cdot (l + \mathcal M)^{O(\mathcal M)} + O(2^{\mathcal M}\, m\, l^3) = O(m) \cdot (l + \mathcal M)^{O(\mathcal M)}.
\end{equation*}

\medskip
\noindent\textbf{Quantum preparation.} A $q$-ary MPS with $r \le m$ sites, open 
boundary conditions, and bond dimension $D$ can be prepared by a quantum circuit 
running in time $O(m\, q \cdot \mathrm{poly}(D))$~\cite{schon2005sequential, melnikov2023quantum}. 
Setting $q = 2^{\mathcal M}$ and $D = l + 1$ gives
\begin{equation*}
T_{\mathrm{prepare}} = O(m \cdot 2^{\mathcal M} \cdot \mathrm{poly}(l)).
\end{equation*}

When $\mathcal M = O(1)$, $(l + \mathcal M)^{O(\mathcal M)} = \mathrm{poly}(l)$ and $2^{\mathcal M} = O(1)$, both 
runtimes are polynomial in $m$, $l$, and $n$.
\end{proof}


\begin{Rem}
The bond dimension $D = l + 1$ is the same as in Theorem~\ref{thm:ref-state} 
(commuting case). The cost of noncommutativity appears instead in:
\begin{itemize}
\item The \emph{local dimension} $q = 2^{\mathcal M}$, which reflects that one MPS site 
groups together all qubits in a connected component of the anticommutation graph.
\item The \emph{classical pre-processing time} $(l+\mathcal M)^{O(\mathcal M)}$, which scales with 
the size of the largest connected component.
\end{itemize}
When the anticommutation graph has only bounded-size components (i.e., $\mathcal M = O(1)$), 
both factors collapse to constants and the algorithm remains efficient.
\end{Rem}

\begin{thm}\label{thm:hdqi-noncomm}
Let $H=\sum_{i=1}^{m}c_iP_i$ be a Hamiltonian on $n$ qubits with $c_i\in\mathbb R$ and distinct $n$-qubit Pauli operators $\{P_i\}$.
Let $\CPP(x)=\sum_{j=0}^{l}a_jx^j$ be a real polynomial of degree $l$ with $\CPP(H)\ne0$, and suppose $\mathcal D_H^{(l)}$ is a weight-$l$ decoding oracle for $H$.
Let $\mathcal M$ denote the maximum size of connected components in the anticommutation graph of $H$.
Then there exists a quantum algorithm that prepares the state
\begin{equation*}
\rho_{\CPP}(H) = \frac{\CPP^2(H)}{\mathrm{Tr}[\CPP^2(H)]}.
\end{equation*}
The algorithm uses one call to $\mathcal D_H^{(l)}$ and the reference state $|R^l_*(H)\rangle$ from Theorem~\ref{thm:ref-state-noncomm}.
The classical pre-processing time is $O(m)\cdot(l+\mathcal M)^{O(\mathcal M)}$, and the quantum runtime is $O(m\cdot2^{\mathcal M}\cdot\mathrm{poly}(l,n))$ plus the cost of one oracle call.
In particular, the algorithm runs in polynomial time when $\mathcal M=O(1)$, $l,m=\mathrm{poly}(n)$, and the oracle is efficient.
\end{thm}
\begin{proof}
The construction follows that of Theorem~\ref{thm:hdqi-commuting} with the 
reference state $|R^l(H)\rangle$ replaced by $|R^l_*(H)\rangle$ from 
Theorem~\ref{thm:ref-state-noncomm}. The key observation is that, by 
Lemma~\ref{lem:operator-noncomm-expansion},
\begin{equation*}
\CPP(H) = \sum_{\mathbf{y} \in \mathbb{F}_2^m} \left( \sum_{s=0}^{l} a_s\, \beta_G^{(s)}(\mathbf{y}, \mathbf{c}) \right) P_{\mathbf{y}}, 
\quad P_{\mathbf{y}} = P_1^{y_1} \cdots P_m^{y_m},
\end{equation*}
where $P_{\mathbf{y}}$ is the ordered Pauli product in 
Eq.~\eqref{eq:Py-ordered}. The state evolution proceeds verbatim with the 
following adjustments:
\begin{itemize}
\item In Step 3, the controlled unitary $\prod_{i=1}^{m} C^{(i)}\!P_i^{A \to B}$ 
is applied in decreasing index order. The cumulative action on register $B$ 
controlled by $|\mathbf{y}\rangle_A$ is then precisely $P_{\mathbf{y}}$.
\item In Step 4, write the ordered product as $P_{\mathbf y}=\omega_{\mathbf y}W(B^\top\mathbf y)$, as in the proof of Theorem~\ref{thm:hdqi-commuting}.
The Bell transform gives $\omega_{\mathbf y}|B^\top\mathbf y\rangle$, and the phase is retained until the inverse transform restores $(P_{\mathbf y}\otimes I)|\Phi_n\rangle$.
\end{itemize}
The remaining steps and the complexity analysis are identical to those of 
Theorem~\ref{thm:hdqi-commuting}, with the reference state preparation cost 
given by Theorem~\ref{thm:ref-state-noncomm} (ii).
\end{proof}

\begin{Rem}
As in the commuting case, when the Pauli operators $\{P_i\}$ are linearly 
independent in $\mathbb{F}_2^{2n}$, an efficient decoding oracle exists for any 
$\mathbf{y} \in \mathbb{F}_2^m$, and Theorem~\ref{thm:hdqi-noncomm} applies for 
polynomials of arbitrary degree $l$.
\end{Rem}

\noindent\textbf{Example: A restricted transverse-field Ising model}

To illustrate our main results for noncommuting Hamiltonians,
let us consider a concrete family of noncommuting Hamiltonians on $2n + 1$ 
qubits:
\begin{equation}\label{eq:H1-def}
H_1 = \sum_{i=1}^{2n} Z_i Z_{i+1} + g \sum_{i=1}^{n} X_{2i}, 
\qquad g \in \mathbb{R}.
\end{equation}
This is a variant of the one-dimensional transverse-field Ising (TFI) model in 
which the transverse field acts only on the even sites. As we shall see, this 
restriction ensures that the anticommutation graph of $H_1$ decomposes into 
bounded-size components, which is precisely the regime where HDQI yields a 
polynomial-time algorithm. 

Let us rewrite $H_1 = \sum_{i=1}^{m} c_i P_i$ with $m = 3n$, we identify
\begin{align*}
P_i =& \begin{cases} Z_i Z_{i+1} & \text{if } 1 \le i \le 2n, \\ X_{2(i - 2n)} & \text{if } 2n+1 \le i \le 3n, \end{cases}
\\
c_i =& \begin{cases} 1 & \text{if } 1 \le i \le 2n, \\ g & \text{if } 2n+1 \le i \le 3n. \end{cases}
\end{align*}
The Hamiltonian is genuinely noncommuting: each $X_{2i}$ anticommutes with the 
$ZZ$ terms that share its support, so this Hamiltonian contains noncommuting Pauli operators.

We now establish two key structural properties of $H_1$ that enable efficient 
HDQI: (a) the symplectic code is trivial, ensuring an efficient decoding 
oracle; (b) the anticommutation graph has bounded component size, ensuring an 
efficient reference state preparation.

The symplectic representations $\{\mathrm{symp}(P_i)\}_{i=1}^{3n}$ are linearly 
independent in $\mathbb{F}_2^{2(2n+1)}$. Indeed, the $2n$ vectors 
$\mathrm{symp}(Z_j Z_{j+1})$ lie in the $Z$-type subspace and form the incidence 
matrix of a linear chain, which has full row rank $2n$. The $n$ vectors 
$\mathrm{symp}(X_{2i})$ lie in the $X$-type subspace and are obviously 
independent. Since $Z$-type and $X$-type subspaces intersect trivially, the 
union of $3n$ vectors is linearly independent. Hence the symplectic code of 
$H_1$ is trivial ($k = 0$), and thus a 
weight-$l$ decoding oracle $\mathcal{D}_{H_1}^{(l)}$ exists efficiently for 
any $l$.

For the anticommutation graph: two Pauli operators anticommute iff their 
symplectic inner product is $1$. Direct computation shows that $X_{2i}$ 
anticommutes with exactly two $ZZ$ terms—namely $Z_{2i-1} Z_{2i}$ and 
$Z_{2i} Z_{2i+1}$, which share the qubit $2i$ with $X_{2i}$—and commutes with 
all other operators in $H_1$. Consequently, the anticommutation graph $G$ of 
$H_1$ decomposes into $n$ disjoint stars; the $i$-th star has center $X_{2i}$ 
and leaves $Z_{2i-1} Z_{2i}, Z_{2i} Z_{2i+1}$, as illustrated in 
Fig.~\ref{fig:anticomm-graph}. The maximum component size is therefore 
$\mathcal M := \max_t |V_t| = 3$.

\begin{figure*}
\centering
\includegraphics[width=0.9\linewidth]{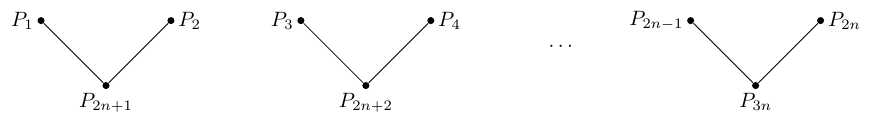}
\caption{Anticommutation graph $G$ of $H_1$ in Eq.~\eqref{eq:H1-def}. The graph 
consists of $n$ disjoint star components, each containing one $X$ term (center) 
and two $ZZ$ terms (leaves). The bounded component size $\CMM = 3$ is the key 
structural property enabling efficient HDQI.}
\label{fig:anticomm-graph}
\end{figure*}

These two structural properties yield the following result:
\begin{prop}[HDQI for the restricted TFI model]\label{prop:ising-example}
Let $g\in\mathbb R$, $n\ge1$, and $H_1$ be as in Eq.~\eqref{eq:H1-def}.
Let $\CPP$ be a real polynomial of degree $l$ with $\CPP(H_1)\ne0$.
The state $\rho_{\CPP}(H_1)=\CPP^2(H_1)/\mathrm{Tr}[\CPP^2(H_1)]$ can be prepared by the algorithm of Theorem~\ref{thm:hdqi-noncomm} with classical and quantum runtimes polynomial in $n$ and $l$.
\end{prop}

\begin{proof}
By the structural analysis above, the symplectic code of $H_1$ is trivial, so 
there exists an efficient weight-$l$ decoding 
oracle $\mathcal{D}_{H_1}^{(l)}$ for any $l$. The anticommutation graph has 
maximum component size $\mathcal M = 3 = O(1)$. Applying 
Theorem~\ref{thm:ref-state-noncomm} (ii) with $\mathcal M = O(1)$ gives 
$T_{\mathrm{pre}} = O(m) \cdot (l + \mathcal M)^{O(\mathcal M)} = \mathrm{poly}(n, l)$ and 
$T_{\mathrm{prepare}} = O(m \cdot 2^{\mathcal M} \cdot \mathrm{poly}(l)) = \mathrm{poly}(n, l)$, 
since $m = 3n$. The complete algorithm of Theorem~\ref{thm:hdqi-noncomm} 
inherits this polynomial scaling.
\end{proof}

Proposition~\ref{prop:ising-example} prepares $\rho_{\CPP}(H_1)$ for any polynomial 
$\CPP(x)$. Combining this with the polynomial approximation of Gibbs states 
established in~\cite[Theorem~3]{schmidhuber2025hamiltonian} yields an explicit 
Gibbs state preparation algorithm.

\begin{cor}[Gibbs state preparation for $H_1$]\label{cor:ising-gibbs}
For any inverse temperature $\beta > 0$ and target accuracy $\delta \in (0, 1)$, 
there exists a univariate polynomial $\CPP(x)$ of degree
\begin{equation}\label{eq:degree-bound}
l \le 1.12\, (2 + |g|)\, \beta\, n + 0.648\, \ln(24 / \delta),
\end{equation}
and a quantum algorithm running in time $\mathrm{poly}(l, n)$, that prepares 
the state $\rho_{\CPP}(H_1)$ satisfying
\begin{equation}\label{eq:gibbs-approx}
\left\| \rho_{\CPP}(H_1) - \frac{\exp(-\beta H_1)}{\mathrm{Tr}\bigl[\exp(-\beta H_1)\bigr]} \right\|_1 \le \delta.
\end{equation}
In particular, when $(2 + |g|)\, \beta = \mathrm{poly}(n)$, the total running 
time is $\mathrm{poly}\bigl(n,\, \ln(1 / \delta)\bigr)$.
\end{cor}
\begin{proof}
The operator norm of $H_1$ satisfies
\begin{equation*}
\|H_1\| \le \left\| \sum_{i=1}^{2n} Z_i Z_{i+1} \right\| + |g| \left\| \sum_{i=1}^{n} X_{2i} \right\| \le 2n + |g|\, n = (2 + |g|)\, n.
\end{equation*}
By~\cite[Theorem~3]{schmidhuber2025hamiltonian}, there exists a polynomial 
$\CPP(x)$ of degree $l \le 1.12\, \beta\, \|H_1\| + 0.648\, \ln(24/\delta)$ such 
that $\rho_{\CPP}(H_1)$ is $\delta$-close in trace distance to the Gibbs state 
$\exp(-\beta H_1) / \mathrm{Tr}[\exp(-\beta H_1)]$. Substituting the bound on 
$\|H_1\|$ gives Eq.~\eqref{eq:degree-bound}. By 
Proposition~\ref{prop:ising-example}, this $\rho_{\CPP}(H_1)$ can be prepared in 
time $\mathrm{poly}(l, n)$.  
\end{proof}

\section{Discussion of Imperfect Decoding}
\label{sec6}
The algorithms in Theorems \ref{thm:hdqi-commuting} and \ref{thm:hdqi-noncomm} assume access to a perfect decoding oracle.  
We now study the effect of replacing it by a decoder that is guaranteed to be correct only on a prescribed set of error labels.  
We first establish a weighted robustness bound for a fixed commuting Hamiltonian and a fixed decoder, and then compare this bound with the more conventional shellwise notion of decoding accuracy.  
We subsequently turn to the noncommuting case.
Now we first give our imperfect decoding model.

\begin{Def}[Imperfect decoding oracle] \label{def:imperfect_decoding_oracle}
Let $B^\top\in\BFF_2^{2n\times m}$ be the parity-check matrix associated with the symplectic code of $H$, and let $\CGG\subseteq\BFF_2^m$.
An imperfect decoding oracle $\CDD_H^{(\CGG)}$ with guaranteed correct set
$\CGG$ is a unitary operator acting as
\begin{equation*}
\CDD_H^{(\CGG)}
\ket{\mathbf y}\ket{\mathbf s}
=
\ket{\mathbf y\oplus\operatorname{Dec}(\mathbf s)}\ket{\mathbf s},
\end{equation*}
where $\operatorname{Dec}:\mathbb F_2^{2n}\to\mathbb F_2^m$ satisfies
\begin{equation*}
\operatorname{Dec}(B^\top\mathbf y)=\mathbf y,
\qquad \forall\,\mathbf y\in\CGG.
\end{equation*}
\end{Def}

\subsection{Commuting Hamiltonians}

We begin with a fixed commuting Hamiltonian and a fixed imperfect decoder.
For such an instance, the relevant error parameter is the total pilot-state weight carried by labels outside the guaranteed correct set, rather than the unweighted fraction of failed labels.

\begin{thm}[Weighted robustness to imperfect decoding]
\label{thm:weighted-robustness-commuting}
Let $H=\sum_{i=1}^m c_iP_i$ be an $n$-qubit Hamiltonian, where $c_i\in\mathbb R$ and $P_1,\ldots,P_m$ are distinct,
mutually commuting $n$-qubit Pauli operators.  Suppose that the Hamiltonian symplectic code of $H$ has minimum distance at least $2l+1$.  
Let $\CPP(x)=\sum_{s=0}^la_sx^s$ be a real polynomial such that $\CPP(H)\ne0$, and let $w_{\mathbf y}$ be the coefficients in Eq.~\eqref{eq:coeff}.  
Let $\rho_{\CPP,\CGG}(H)$ denote the output of the algorithm in Theorem \ref{thm:hdqi-commuting} using $\CDD_H^{(\CGG)}$, and let $\rho_{\CPP}(H)$ denote the output obtained using the perfect oracle.

Define 
\begin{equation*}
\eta_{\rm w}
:=
\frac{
\displaystyle
\sum_{\substack{\mathbf y\notin\CGG\\|\mathbf y|\le l}}
|w_{\mathbf y}|^2
}{
\displaystyle
\sum_{\substack{\mathbf y\in\mathbb F_2^m\\|\mathbf y|\le l}}
|w_{\mathbf y}|^2
}.
\end{equation*}
Then
\begin{equation}\label{260728eq1}
F\bigl(\rho_{\CPP,\CGG}(H),\rho_{\CPP}(H)\bigr)
\ge1-\eta_{\rm w},
\end{equation}
where $F(\rho,\sigma):=\|\sqrt\rho\sqrt\sigma\|_1$ denotes the root fidelity.
\end{thm}

\begin{proof}
Following the proof of Theorem~\ref{thm:hdqi-commuting}, write $P_{\mathbf y}=\omega_{\mathbf y}W(B^\top\mathbf y)$.
After the coherent Bell measurement, the state is
\begin{equation*}
\frac{1}{\mathcal N}
\sum_{\substack{\mathbf y\in\mathbb F_2^m\\|\mathbf y|\le l}}
w_{\mathbf y}\omega_{\mathbf y}
\ket{\mathbf y}_A\ket{B^\top\mathbf y}_{BC},
\qquad
\mathcal N^2=
\sum_{\substack{\mathbf y\in\mathbb F_2^m\\|\mathbf y|\le l}}
|w_{\mathbf y}|^2.
\end{equation*}
Using the perfect oracle and then undoing the coherent Bell measurement gives the purification
\begin{equation}\label{260728eq2}
\ket{\psi_{\CPP}(H)}
=
\frac{1}{\mathcal N}
\sum_{|\mathbf y|\le l}
w_{\mathbf y}\ket{\mathbf0}_A
\otimes(P_{\mathbf y}\otimes I)\ket{\Phi_n}_{BC}.
\end{equation}
Using the imperfect oracle instead gives
\begin{equation}\label{260728eq3}
\ket{\psi_{\CPP,\CGG}(H)}
=
\frac{1}{\mathcal N}
\sum_{|\mathbf y|\le l}
w_{\mathbf y}
\ket{\mathbf y\oplus\mathrm{Dec}(B^\top\mathbf y)}_A
\otimes(P_{\mathbf y}\otimes I)\ket{\Phi_n}_{BC}.
\end{equation}
Tracing out registers $AC$ gives $\rho_{\CPP}(H)$ and $\rho_{\CPP,\CGG}(H)$, respectively.

The distance assumption implies that, for $|\mathbf x|,|\mathbf y|\le l$,
\begin{equation}
\langle\Phi_n|
(P_{\mathbf x}^\dagger P_{\mathbf y}\otimes I)
|\Phi_n\rangle
=\delta_{\mathbf x,\mathbf y}.
\end{equation}
Therefore, Uhlmann's theorem and Eqs. \eqref{260728eq2} and \eqref{260728eq3} give
\begin{align*}
F\bigl(\rho_{\CPP,\CGG}(H),\rho_{\CPP}(H)\bigr)
&\ge
\bigl|\langle\psi_{\CPP}(H)|
\psi_{\CPP,\CGG}(H)\rangle\bigr|
\\
&=
\frac{1}{\mathcal N^2}
\sum_{|\mathbf y|\le l}
|w_{\mathbf y}|^2
\langle\mathbf0|
\mathbf y\oplus\mathrm{Dec}(B^\top\mathbf y)\rangle\\
&\ge
\frac{1}{\mathcal N^2}
\sum_{\substack{\mathbf y\in\CGG\\|\mathbf y|\le l}}
|w_{\mathbf y}|^2\\
&=1-\eta_{\rm w}.
\end{align*}
\end{proof}

Theorem \ref{thm:weighted-robustness-commuting} gives a deterministic bound in terms of the actual reference-state weight of the decoder failure set.  
We next ask whether this weighted quantity can be controlled by assuming that the decoder is correct on most labels in each Hamming-weight shell.

For $0\le j\le m$, we denote
\begin{equation*}
\mathcal S_j:=\{\mathbf y\in\mathbb F_2^m:|\mathbf y|=j\},
\qquad
\CGG_j:=\CGG\cap\mathcal S_j.
\end{equation*}
Fix $0\le\epsilon<1$.  
We assume that $\mathbf 0\in\CGG$, so that $\operatorname{Dec}(\mathbf 0)=\mathbf 0$, and that, for every
$1\le j\le \min\{l,m\}$,
\begin{equation}\label{260718eq1}
|\CGG_j|\ge (1-\epsilon)|\mathcal S_j|.
\end{equation}

\begin{Rem}
The shellwise condition \eqref{260718eq1} controls the
fraction of correctly decoded labels with respect to the uniform counting measure on each shell.  
In contrast, the overlap in the preceding proof
weights the label $\mathbf y$ by $|w_{\mathbf y}|^2/\mathcal N^2$.  
Consequently, the shellwise condition does not by itself control $\eta_{\rm w}$: a small set of failed labels may carry a large fraction of the reference-state probability.  This is the reason that the natural general statement is Eq. \eqref{260728eq1}, rather than a bound depending only on $\epsilon$.
\end{Rem}

On the other hand, one might hope that the additional assumption
\begin{equation}\label{260729eq1}
0<c_{\min}\le |c_i|\le c_{\max}, 
\end{equation}
forces the weights $|w_{\mathbf y}|^2$ to be sufficiently uniform inside each Hamming shell.  
The following example shows that this is false for some polynomial $\CPP$.  
The obstruction is cancellation among the contributions of different powers of $H$.

\begin{Examp}\label{ex:commuting-case} 
Fix any $\Delta>1$ and choose an irrational number $\alpha\in(1,\Delta)$.  
On $m$ qubits, let
\begin{equation*}
H_\alpha:=\alpha Z_1+\sum_{i=2}^m Z_i.
\end{equation*}
So $ c_{\max}/c_{\min} = \alpha<\Delta$.  
The spectrum of $H_\alpha$ consists of the $2m$ distinct numbers in the following set
\begin{equation*}
\{\alpha s_1+t: \;
s_1\in\{\pm1\}, 
t\in\{-(m-1),-(m-3),\ldots,m-1\} \}.
\end{equation*}
The irrationality of $\alpha$ guarantees that these numbers are pairwise distinct.  
Hence there exists a real interpolation polynomial $\CPP_m$ of degree at most $2m-1$ such that
\begin{equation}
\CPP_m(\alpha+t)=2,
\qquad
\CPP_m(-\alpha+t)=0
\end{equation}
for every $t \in\{-(m-1),-(m-3),\ldots,m-1\}$.  
So we have
\begin{equation}\label{260729eq3}
\CPP_m(H_\alpha)=I+Z_1.
\end{equation}

The Pauli operators $Z_1,\ldots,Z_m$ are symplectically independent, so a (perfect) decoder of arbitrary weight exists.  
Let $\CGG=\mathbb F_2^m\setminus\{\mathbf e_1\}$.  
Let the imperfect decoder fail only on $\mathbf e_1$, with
$\mathrm{Dec}(B^\top\mathbf e_1)=\mathbf0$, and decode every other error correctly.  
On the weight-one shell, its success fraction is
\begin{equation*}
\frac{m-1}{m}=1-\frac1m,
\end{equation*}
and it is perfect on every other shell.  
Thus
Eq. \eqref{260718eq1} holds with $\epsilon=1/m\to0$.
Nevertheless, Eq. \eqref{260729eq3} gives
\begin{equation}
 w_{\mathbf0}=w_{\mathbf e_1}=1,
 \qquad
 w_{\mathbf y}=0
 \quad\text{for all other }\mathbf y,
\end{equation}
so that $\eta_{\rm w}=1/2$.  


We now compute the imperfect-decoding output explicitly.
Let $l_m:=\deg(\CPP_m)$.  
Since $\CPP_m$ has $m$ distinct zeros and is not identically zero, $m\le l_m\le 2m-1$.  
Since
\begin{equation*}
\CPP_m(H_\alpha)=I+Z_1,
\end{equation*}
the normalization factor is $\mathcal N=\sqrt2$, and the reference state is
\begin{equation*}
\ket{R^{l_m}(H_\alpha)}_A
=
\frac{1}{\sqrt2}
\left(
\ket{\mathbf0}_A+
\ket{\mathbf e_1}_A
\right).
\end{equation*}
After preparing the maximally entangled state on registers $B$ and $C$, the joint state is 
\begin{equation*}
\ket{\psi_1}_{ABC}
=
\frac{1}{\sqrt2}
\left(
\ket{\mathbf0}_A+
\ket{\mathbf e_1}_A
\right)
\otimes\ket{\Phi_m}_{BC}.
\end{equation*}
After the controlled Pauli operations, we obtain
\begin{align*}
\ket{\psi_2}_{ABC}
=
\frac{1}{\sqrt2}
\bigl(
\ket{\mathbf0}_A\otimes\ket{\Phi_m}_{BC}
+
\ket{\mathbf e_1}_A
\otimes(Z_1\otimes I)\ket{\Phi_m}_{BC}
\bigr).
\end{align*}
The coherent Bell measurement maps
\begin{align*}
\ket{\Phi_m}_{BC}
&\longmapsto\ket{\mathbf0}_{BC},\\
(Z_1\otimes I)\ket{\Phi_m}_{BC}
&\longmapsto\ket{B^\top\mathbf e_1}_{BC}.
\end{align*}
Therefore,
\begin{align}
\ket{\psi_3}_{ABC}
=
\frac{1}{\sqrt2}
\bigl(
\ket{\mathbf0}_A\ket{\mathbf0}_{BC}
+
\ket{\mathbf e_1}_A
\ket{B^\top\mathbf e_1}_{BC}
\bigr).
\end{align}

We now apply the imperfect decoding oracle.  
By convention,
$\operatorname{Dec}(\mathbf0)=\mathbf0$, whereas the decoder was chosen so that
\begin{equation*}
\operatorname{Dec}(B^\top\mathbf e_1)=\mathbf0.
\end{equation*}
Consequently,
\begin{align*}
\CDD_H^{(\CGG)}
\ket{\mathbf0}_A\ket{\mathbf0}_{BC}
&=
\ket{\mathbf0}_A\ket{\mathbf0}_{BC},\\
\CDD_H^{(\CGG)}
\ket{\mathbf e_1}_A
\ket{B^\top\mathbf e_1}_{BC}
&=
\ket{\mathbf e_1\oplus
\operatorname{Dec}(B^\top\mathbf e_1)}_A
\ket{B^\top\mathbf e_1}_{BC}\\
&=
\ket{\mathbf e_1}_A
\ket{B^\top\mathbf e_1}_{BC},
\end{align*}
so applying the imperfect decoding oracle does not change the state.
Tracing out register $A$, we get
\begin{align*}
\psi_4^{\mathrm{imp}}
&=
\operatorname{Tr}_A
\left[
\ket{\psi_3}
\!\bra{\psi_3}
\right]
=
\frac12\ket{\mathbf0} \bra{\mathbf0}_{BC}
+
\frac12 
\ket{B^\top\mathbf e_1} \bra{B^\top\mathbf e_1}_{BC} .
\end{align*}
Applying the inverse of the coherent Bell measurement gives
\begin{align*}
\psi_5^{\mathrm{imp}}
=
\frac{1}{2}
\bigl( \ket{\Phi_m}\bra{\Phi_m}_{BC}
+
(Z_1\otimes I)\ket{\Phi_m}\!\bra{\Phi_m}_{BC}(Z_1\otimes I)
\bigr).
\end{align*}
So the imperfect-decoding output is
\begin{align*}
\rho_{\CPP_m,\CGG}(H_\alpha)
=
\operatorname{Tr}_C\left[\psi_5^{\mathrm{imp}}\right]
=
\frac12\frac{I}{2^m}
+
\frac12 Z_1\frac{I}{2^m}Z_1
=
\frac{I}{2^m}.
\end{align*}

For comparison, the perfect decoder maps both labels in register $A$ to $\mathbf0$.  
Thus
\begin{align}
\ket{\psi_4^{\mathrm{perf}}}_{ABC}
=
\frac{1}{\sqrt2}\ket{\mathbf0}_A\otimes
\left(
\ket{\mathbf0}_{BC}
+
\ket{B^\top\mathbf e_1}_{BC}
\right),
\end{align}
and, after undoing the coherent Bell measurement,
\begin{align*}
\ket{\psi_5^{\mathrm{perf}}}_{ABC}
&=
\frac{1}{\sqrt2}\ket{\mathbf0}_A\otimes
\left(
\ket{\Phi_m}_{BC}
+
(Z_1\otimes I)\ket{\Phi_m}_{BC}
\right) \\
&=
\frac{1}{\sqrt2}\ket{\mathbf0}_A\otimes
\bigl((I+Z_1)\otimes I\bigr)\ket{\Phi_m}_{BC}.
\end{align*}
Consequently, 
the perfect-decoding output is
\begin{equation*}
\rho_{\CPP_m}(H_\alpha)
=
\frac{(I+Z_1)^2}{\operatorname{Tr}[(I+Z_1)^2]}
=
\frac{I+Z_1}{2^m}.
\end{equation*}
Therefore,
\begin{align}
F\bigl(\rho_{\CPP_m,\CGG}(H_\alpha),
\rho_{\CPP_m}(H_\alpha)\bigr)
&=
2^{m-1}
\sqrt{\frac{1}{2^m}\frac{1}{2^{m-1}}}
=
\frac{1}{\sqrt2}.
\end{align}

Thus, for every fixed $\Delta>1$, there is a family for which the shellwise failure fraction tends to zero while the fidelity remains bounded away from one.  
In particular, no lower bound of the form
\begin{align*}
F\bigl(\rho_{\CPP,\CGG}(H),\rho_{\CPP}(H)\bigr)
&\ge f(\epsilon,\Delta),\\
\text{where} f(\epsilon,\Delta)&\longrightarrow1
\quad\text{as }\epsilon\longrightarrow0,
\end{align*}
can hold uniformly over all system sizes and all polynomial degrees.
\end{Examp}

The preceding example rules out a robustness estimate depending only on the shellwise failure fraction, even when the coefficient dynamic range is bounded.  
For commuting Hamiltonians, however, an $\epsilon$-dependent estimate is recovered when the coefficients are independent and identically distributed.

\begin{thm}
\label{cor:iid-coefficient-robustness}
Fix distinct, mutually commuting $n$-qubit Pauli operators $P_1,\ldots,P_m$ whose associated Hamiltonian symplectic code has minimum distance at least $2l+1$.  
Let $\CPP(x)=\sum_{s=0}^l a_sx^s$ be a real univariate polynomial of degree at most $l$, and fix an imperfect decoding oracle $\CDD_H^{(\CGG)}$ satisfying Eq. \eqref{260718eq1}.
Let $c_1,\ldots,c_m$ be independent random variables with a common probability distribution $\mu$ on $\mathbb R$, and define
\begin{equation*}
\mathbf c:=(c_1,\ldots,c_m),
\qquad
H(\mathbf c):=\sum_{i=1}^m c_iP_i.
\end{equation*}
Assume that $\CPP(H(\mathbf c))\ne0$ almost surely.  
For each realization with $\CPP(H(\mathbf c))\ne0$, let $\rho_{\CPP,\CGG}(H(\mathbf c))$ and $\rho_{\CPP}(H(\mathbf c))$ denote the outputs obtained using the imperfect and perfect decoding oracles, respectively.  
Then
\begin{equation}\label{eq:iid-coefficient-average-fidelity}
\mathbb E_{\mathbf c\sim\mu^{\otimes m}}
\left[
F\bigl(
\rho_{\CPP,\CGG}(H(\mathbf c)),
\rho_{\CPP}(H(\mathbf c))
\bigr)
\right]
\ge 1-\epsilon.
\end{equation}
\end{thm}

\begin{proof}
Fix a realization of $\mathbf c$ such that $\CPP(H(\mathbf c))\ne0$.
Let $\mathfrak S_m$ denote the symmetric group on $\{1,\ldots,m\}$, and for $\pi\in\mathfrak S_m$ define
\[
\mathbf c^\pi:=(c_{\pi(1)},\ldots,c_{\pi(m)}).
\]
The matrix $B^\top$ and the decoding oracle remain fixed as the coefficients are permuted. 
Let $w_{\mathbf y}$ be the coefficient in Eq. \eqref{eq:coeff} for $\mathbf c$.
For $\pi\in\mathfrak S_m$, let $w_{\mathbf y}^{(\pi)}$ be the coefficient in Eq. \eqref{eq:coeff} obtained from $\mathbf c^\pi$.
With the action
\[
(\pi\mathbf y)_i:=y_{\pi^{-1}(i)},
\]
Eq. \eqref{eq:coeff} gives
\begin{equation}\label{260718eq2}
w_{\mathbf y}^{(\pi)}=w_{\pi\mathbf y}.
\end{equation}
Consequently, for every $j$ and every fixed $\mathbf y\in\mathcal S_j$, the multiset $\{|w_{\mathbf y}^{(\pi)}|^2:\pi\in\mathfrak S_m\}$ contains every value $|w_{\mathbf z}|^2$, $\mathbf z\in\mathcal S_j$, equally often.  
In particular, let 
\[ W_j:=\sum_{\mathbf z\in\mathcal S_j}|w_{\mathbf z}|^2,\]
then
\begin{equation}\label{260718eq3}
\frac{1}{m!}\sum_{\pi\in\mathfrak S_m}
\sum_{\mathbf y\in\CGG_j}|w_{\mathbf y}^{(\pi)}|^2
=
\frac{|{\CGG}_j|}{|{\mathcal S}_j|}W_j
\ge
(1-\epsilon)W_j
\end{equation}
for $1\le j\le\min\{l,m\}$.  
Equation \eqref{260718eq2} also shows that the number
\[
\mathcal N^2
:=
\sum_{\substack{\mathbf y\in\mathbb F_2^m\\|\mathbf y|\le l}}
|w_{\mathbf y}^{(\pi)}|^2
=
\sum_{j=0}^{\min\{l,m\}}W_j
\]
is independent of $\pi$. 
The distance assumption implies that the operators $P_{\mathbf y}$ with $|\mathbf y|\le l$ are orthogonal under the Hilbert-Schmidt inner product. 
Since $\CPP(H(\mathbf c))\ne0$ and $\mathcal N$ is independent of $\pi$,
\[
\mathcal N^2=2^{-n}\operatorname{Tr}\!\left[\CPP^2\bigl(H(\mathbf c^\pi)\bigr)\right]>0
\]
for every $\pi$, so $\CPP(H(\mathbf c^\pi))\ne0$ for every $\pi$. 
Define
\[
\rho_{\CPP}^{(\pi)}:=\rho_{\CPP}(H(\mathbf c^\pi)),\qquad
\rho_{\CPP,\CGG}^{(\pi)}:=\rho_{\CPP,\CGG}(H(\mathbf c^\pi)).
\]

Following the proof of Theorem \ref{thm:hdqi-commuting}, the perfect and imperfect outputs admit the purifications
\begin{align*}
\ket{\psi_{\CPP}^{(\pi)}}
&=
\frac{1}{\mathcal N}
\sum_{|\mathbf y|\le l}
w_{\mathbf y}^{(\pi)}
\ket{\mathbf 0}_A
\otimes(P_{\mathbf y}\otimes I)\ket{\Phi_n}_{BC}, \\
\ket{\psi_{\CPP,\CGG}^{(\pi)}}
&=
\frac{1}{\mathcal N}
\sum_{|\mathbf y|\le l}
w_{\mathbf y}^{(\pi)}
\ket{\mathbf y\oplus\operatorname{Dec}(B^\top\mathbf y)}_A
\otimes(P_{\mathbf y}\otimes I)\ket{\Phi_n}_{BC}.
\end{align*}
For $|\mathbf x|,|\mathbf y|\le l$, the distance assumption implies
\[
\langle\Phi_n|
(P_{\mathbf x}^\dagger P_{\mathbf y}\otimes I)
|\Phi_n\rangle
=
\delta_{\mathbf x,\mathbf y}.
\]
Hence, by Uhlmann's theorem,
\begin{equation*}
F\bigl(\rho_{\CPP,\CGG}^{(\pi)},\rho_{\CPP}^{(\pi)}\bigr)
\ge
\bigl|\langle\psi_{\CPP}^{(\pi)}|
\psi_{\CPP,\CGG}^{(\pi)}\rangle\bigr| 
\ge
\frac{1}{\mathcal N^2}
\sum_{\substack{\mathbf y\in\CGG\\|\mathbf y|\le l}}
|w_{\mathbf y}^{(\pi)}|^2.
\end{equation*}
Averaging over all $\pi$ and using Eq.~\eqref{260718eq3}, we obtain
\begin{equation}\label{260604eq3}
\begin{aligned}
\frac{1}{m!}\sum_{\pi\in\mathfrak S_m}
F\bigl(\rho_{\CPP,\CGG}^{(\pi)},\rho_{\CPP}^{(\pi)}\bigr)
&\ge
\frac{1}{\mathcal N^2}
\left(
W_0+
\sum_{j=1}^{\min\{l,m\}}
\frac{|\CGG_j|}{|\mathcal S_j|}W_j
\right)  \ge
\frac{W_0+(1-\epsilon)
\sum_{j=1}^{\min\{l,m\}}W_j}{\mathcal N^2} \ge 1-\epsilon.
\end{aligned}
\end{equation}

In addition, for every $\pi$ the Fuchs-van de Graaf inequality gives
\begin{equation*}
\bigl\|\rho_{\CPP,\CGG}^{(\pi)}-\rho_{\CPP}^{(\pi)}\bigr\|_1
\le 2\sqrt{1-F\bigl(\rho_{\CPP,\CGG}^{(\pi)},\rho_{\CPP}^{(\pi)}\bigr)^2}.
\end{equation*}
Since $x\mapsto\sqrt{1-x^2}$ is concave and decreasing on $[0,1]$, Jensen's inequality and Eq. \eqref{260604eq3} imply
\begin{equation}\label{260604eq2}
\begin{aligned}
\frac{1}{m!}\sum_{\pi\in\mathfrak S_m}
\bigl\|\rho_{\CPP,\CGG}^{(\pi)}-\rho_{\CPP}^{(\pi)}\bigr\|_1
&\le
2\sqrt{1- \left(\frac{1}{m!}\sum_{\pi\in\mathfrak S_m}F\bigl(\rho_{\CPP,\CGG}^{(\pi)},\rho_{\CPP}^{(\pi)}\bigr) \right)^2}  \le 2\sqrt{1-(1-\epsilon)^2} \le 2\sqrt{2\epsilon}.
\end{aligned}
\end{equation}

Define
\begin{equation*}
F_{\mathbf c^\pi}
:=
F\bigl(
\rho_{\CPP,\CGG}(H(\mathbf c^\pi)),
\rho_{\CPP}(H(\mathbf c^\pi))
\bigr).
\end{equation*}
Since $c_1,\ldots,c_m$ are independent and identically distributed,
$\mathbf c^\pi$ has the same distribution as $\mathbf c$ for every $\pi\in\mathfrak S_m$.  
Therefore,
\begin{align*}
\mathbb E_{\mathbf c\sim\mu^{\otimes m}}[F_{\mathbf c}]
=
\frac{1}{m!}
\sum_{\pi\in\mathfrak S_m}
\mathbb E_{\mathbf c\sim\mu^{\otimes m}}[F_{\mathbf c^\pi}]
=
\mathbb E_{\mathbf c\sim\mu^{\otimes m}}
\left[
\frac{1}{m!}
\sum_{\pi\in\mathfrak S_m}
F_{\mathbf c^\pi}
\right].
\end{align*}
By Eq. \eqref{260604eq3}, the quantity inside the expectation is at least
$1-\epsilon$ for every realization for which $\CPP(H(\mathbf c))\ne0$. 
Since this holds almost surely, taking the expectation proves Eq. \eqref{eq:iid-coefficient-average-fidelity}.
\end{proof}

\subsection{Noncommuting Hamiltonians}

For noncommuting Hamiltonians, the reference-state coefficients depend not only on the coefficients $c_i$, but also on which pairs of Pauli operators anticommute. 
Consequently, permuting the coefficients while keeping the Pauli operators fixed does not simply permute the reference-state amplitudes within each Hamming-weight shell. 
Thus, the permutation-averaging argument for commuting Hamiltonians does not extend to the noncommuting case directly.  
The following counterexample shows that the bound $F\ge1-\epsilon$ can fail even after averaging over coefficient permutations.

\begin{Rem}\label{260722rem1}
We give an explicit example showing that shellwise decoding accuracy does not imply the bound $F\ge1-\epsilon$ for a general noncommuting Hamiltonian, even after averaging over coefficient permutations.
We remark that this is also a counterexample to Theorem 23 in HDQI \cite{schmidhuber2025hamiltonian}.

Take $n=m=100$, $l=2$, and $\epsilon=1/1000$.   
Let $G$ be the complete graph on $\{1,\ldots,100\}$ with the two edges $\{1,2\}$ and $\{3,4\}$ removed.  
Recall that an edge of $G$ represents anticommutation of Paulis.  
Thus, $P_1$ commutes with $P_2$, $P_3$ commutes with $P_4$, and every other pair of distinct Pauli operators anticommutes.  
Such a family can be realized explicitly on 100 qubits by setting
\begin{equation}\label{260724eq1}
P_i
:= X_i
\prod_{\substack{1\le r<i\\ \{r,i\}\in E(G)}} Z_r,
\qquad 1\le i\le 100.
\end{equation}
Indeed, for $i<j$, the operator $P_j$ contains $Z_i$ if and only if $\{i,j\}\in E(G)$, and hence $P_iP_j=-P_jP_i$ exactly when $\{i,j\}\in E(G)$. 
Moreover, the $X$-part of $\operatorname{symp}(P_i)$ is the $i$-th standard basis vector.  
Therefore,
the vectors $\operatorname{symp} (P_1), \ldots,\operatorname{symp}(P_{100})$ are
linearly independent, so
$\ker(B^\top)=\{\mathbf 0\}$.

Consider
\begin{equation} 
H=\sum_{i=1}^{100} P_i,
\qquad
\CPP(x)=x^2-100.
\end{equation}
Since $P_i^2=I$, we have
\begin{align}
\begin{aligned}
\CPP(H) &=
100 I+
\sum_{1\le i<j\le 100}
(P_iP_j+P_jP_i) -100 I = 2P_1P_2+2P_3P_4.
\end{aligned}
\label{260725eq2}
\end{align} 
By the noncommuting polynomial expansion~\eqref{eq:PH-noncomm-expansion}, the amplitudes of the reference state~\eqref{eq:ref-state-star} are the coefficients associated with the ordered-product labels in $\CPP(H)$.  
Equation \eqref{260725eq2}
therefore shows that the only nonzero amplitudes are those of $\mathbf e_1\oplus\mathbf e_2$ and $\mathbf e_3\oplus\mathbf e_4$, both equal to $2$.  
Hence $\mathcal N^2=2^2+2^2=8$, and the normalized pilot state is
\begin{equation}\label{eq:thm23-counterexample-pilot}
\ket{R_*^2(H)}_A
=
\frac{1}{\sqrt2}
\left(
\ket{\mathbf e_1\oplus\mathbf e_2}_A
+
\ket{\mathbf e_3\oplus\mathbf e_4}_A
\right).
\end{equation}
Thus, the entire pilot-state probability is supported on two errors in the weight-two shell.

We now define an imperfect decoder with shellwise failure fraction at most $\epsilon=1/1000$.
On the weight-one shell, let it decode perfectly.  
On the weight-two shell, let it fail only on
\begin{equation}\label{260725eq3}
\mathbf e_1\oplus\mathbf e_2,
\quad
\mathbf e_3\oplus\mathbf e_4.
\end{equation}
For these failed errors, let the decoder leave the first register
unchanged:
\begin{equation}
\CDD_H^{(\CGG)}
\ket{\mathbf y}_A\ket{B^\top\mathbf y}_{BC}
=
\ket{\mathbf y}_A\ket{B^\top\mathbf y}_{BC}.
\end{equation}
For all remaining errors of weight one or two, let it decode perfectly:
\begin{equation}
\CDD_H^{(\CGG)}
\ket{\mathbf y}_A\ket{B^\top\mathbf y}_{BC}
=
\ket{\mathbf 0}_A\ket{B^\top\mathbf y}_{BC}.
\end{equation}
Because $B^\top$ is injective, this action on the indicated basis states
extends to a unitary.  
The decoder is correct on all weight-one errors and on $4948$ of the $4950$ weight-two errors, so in both shells its failure fraction
\begin{equation}
 \le 1- 4948/4950 < 1/1000 = \epsilon.
\end{equation} 

We now compute the output fidelity using the HDQI algorithm.
After preparing the pilot state and the
maximally entangled state, the joint state is
\begin{equation} 
\ket{\psi_1}_{ABC}
=
\frac{1}{\sqrt2}
\left(
\ket{\mathbf e_1\oplus\mathbf e_2}_A
+
\ket{\mathbf e_3\oplus\mathbf e_4}_A
\right)
\otimes\ket{\Phi_n}_{BC}.
\end{equation}
After the controlled Pauli operations, the state becomes
\begin{align}
\ket{\psi_2}_{ABC}
=
\frac{1}{\sqrt2}
\bigl(&
\ket{\mathbf e_1\oplus\mathbf e_2}_A
\otimes(P_1P_2\otimes I)\ket{\Phi_n}_{BC} +
\ket{\mathbf e_3\oplus\mathbf e_4}_A
\otimes(P_3P_4\otimes I)\ket{\Phi_n}_{BC}
\bigr).
\end{align}
For the Pauli operators in Eq.~\eqref{260724eq1}, we have $P_1P_2=X_1X_2$ and $P_3P_4=X_3X_4$.
Thus both phases in the Bell transform from the proof of Theorem~\ref{thm:hdqi-commuting} are $1$, and the coherent Bell measurement gives
\begin{align} 
\ket{\psi_3}_{ABC}
=
\frac{1}{\sqrt2}
\bigl(&
\ket{\mathbf e_1\oplus\mathbf e_2}_A
\ket{B^\top(\mathbf e_1\oplus\mathbf e_2)}_{BC} +
\ket{\mathbf e_3\oplus\mathbf e_4}_A
\ket{B^\top(\mathbf e_3\oplus\mathbf e_4)}_{BC}
\bigr).
\end{align}
Up to this point, the perfect and imperfect procedures are identical.

For the perfect decoder, register $A$ is uncomputed, so
\begin{align}
\ket{\psi_{4}^{\mathrm{perf}}}_{ABC}
=
\frac{1}{\sqrt2}\ket{\mathbf 0}_A\otimes
\bigl(&
\ket{B^\top(\mathbf e_1\oplus\mathbf e_2)}_{BC} +
\ket{B^\top(\mathbf e_3\oplus\mathbf e_4)}_{BC}
\bigr).
\end{align}
Undoing the coherent Bell measurement gives
\begin{align} 
\ket{\psi_{5}^{\mathrm{perf}}}_{ABC}
=
\frac{1}{\sqrt2}\ket{\mathbf 0}_A\otimes
\bigl(&
(P_1P_2\otimes I)\ket{\Phi_n}_{BC} +
(P_3P_4\otimes I)\ket{\Phi_n}_{BC}
\bigr).
\end{align}
Finally, the perfect-decoding output is
\begin{align}\label{260726eq3}
\rho_{\CPP}(H)
&=
\operatorname{Tr}_{AC}
\left[
\ket{\psi_{5}^{\mathrm{perf}}}
\!\bra{\psi_{5}^{\mathrm{perf}}}
\right] \\
&=
\frac{1}{2^{n+1}}
\left(
2I+P_1P_2P_3P_4+P_3P_4P_1P_2
\right)
\nonumber\\
&=
\frac{I+P_1P_2P_3P_4}{2^n}. \nonumber
\end{align}

For the imperfect decoder, both nonzero pilot-state components belong to the failure set \eqref{260725eq3}.  
Therefore, the decoder leaves their labels in register $A$, and
\begin{equation}
\ket{\psi_{4}^{\mathrm{imp}}}_{ABC}
=
\ket{\psi_3}_{ABC}.
\end{equation}
After undoing the coherent Bell measurement, the state is
\begin{align} 
\ket{\psi_{5}^{\mathrm{imp}}}_{ABC}
=
\frac{1}{\sqrt2}
\bigl(&
\ket{\mathbf e_1\oplus\mathbf e_2}_A
\otimes(P_1P_2\otimes I)\ket{\Phi_n}_{BC}
 +
\ket{\mathbf e_3\oplus\mathbf e_4}_A
\otimes(P_3P_4\otimes I)\ket{\Phi_n}_{BC}
\bigr).
\end{align}
The two states of register $A$ are orthogonal.  
Consequently, tracing out register $A$ removes the two cross terms, and the imperfect-decoding output is
\begin{align}\label{260726eq1}
\rho_{\CPP,\CGG}(H)
=
\operatorname{Tr}_{AC}
\left[
\ket{\psi_{5}^{\mathrm{imp}}}
\!\bra{\psi_{5}^{\mathrm{imp}}}
\right] 
=
\frac{I}{2^n}.
\end{align} 

Finally, the fidelity of $\rho_{\CPP,\CGG}(H)$ and $\rho_{\CPP}(H)$ is 
\begin{equation} 
\mathcal F
\bigl(\rho_{\CPP,\CGG}(H),\rho_{\CPP}(H)\bigr)
=
2^{n-1}
\sqrt{\frac{1}{2^n}\frac{2}{2^n}}
=
\frac{1}{\sqrt2}.
\end{equation}
Since all coefficients are equal to one, this value is unchanged after averaging over coefficient permutations. 
\end{Rem}

The counterexample motivates averaging over the decoding map rather than over the coefficients.  
We now formulate the corresponding decoder-averaged robustness statement.

Given $H$, $l$, and $0\le\epsilon<1$, let $L:=\min\{l,m\}$. Let $\mathfrak D_\epsilon$ be the set of all decoding maps
$\operatorname{Dec}:\mathbb F_2^{2n}\to\mathbb F_2^m$ satisfying
\begin{align}
\operatorname{Dec}(\mathbf 0)
&=\mathbf 0,
\nonumber\\
\left|
\left\{
\mathbf y\in\mathcal S_j:
\operatorname{Dec}(B^\top\mathbf y)=\mathbf y
\right\}
\right|
&\ge(1-\epsilon)|\mathcal S_j|,
\label{260722eq1}
\end{align}
for every $1\le j\le L$.  
For
$\operatorname{Dec}\in\mathfrak D_\epsilon$, let $\CDD_H^{(\operatorname{Dec})}$ be the coherent decoding oracle
\begin{equation*}
\CDD_H^{(\operatorname{Dec})}
\ket{\mathbf y}\ket{\mathbf s}
=
\ket{\mathbf y\oplus\operatorname{Dec}(\mathbf s)}
\ket{\mathbf s}.
\end{equation*}

\begin{prop}
\label{prop:decoder-averaged-robustness-noncomm}
Let $H=\sum_{i=1}^m c_iP_i$ be an $n$-qubit Hamiltonian with $c_i\in\mathbb R$ and distinct Pauli operators $P_i$.  
Suppose that the symplectic code $\mathcal C^\perp=\ker(B^\top)$ has minimum distance $d^\perp\ge2l+1$.  
Let $\CPP(x)=\sum_{s=0}^la_sx^s$ be a real polynomial such that $\CPP(H)\ne0$.
For each $\operatorname{Dec}\in\mathfrak D_\epsilon$, let
$\rho_{\CPP,\operatorname{Dec}}(H)$ denote the output of the algorithm in Theorem \ref{thm:hdqi-noncomm} using $\CDD_H^{(\operatorname{Dec})}$.  
Then
\begin{equation*} 
\frac{1}{|\mathfrak D_\epsilon|}
\sum_{\operatorname{Dec}\in\mathfrak D_\epsilon}
F\bigl(
\rho_{\CPP,\operatorname{Dec}}(H),
\rho_{\CPP}(H)
\bigr)
\ge1-\epsilon,
\end{equation*}
where
$F(\rho,\sigma):=\|\sqrt\rho\sqrt\sigma\|_1$ denotes the fidelity,
and
\begin{align}\label{260722eq3}
\frac{1}{|\mathfrak D_\epsilon|}
\sum_{\operatorname{Dec}\in\mathfrak D_\epsilon}
\bigl\|
\rho_{\CPP,\operatorname{Dec}}(H)-\rho_{\CPP}(H)
\bigr\|_1 
&\le2\sqrt{2\epsilon}.
\end{align}
\end{prop}

\begin{proof}
The distance assumption implies that the syndrome map $\mathbf y\mapsto B^\top\mathbf y$ is injective on $\{\mathbf y:|\mathbf y|\le l\}$.  
In particular,
$\mathfrak D_\epsilon$ is nonempty.

We first record the symmetry of the uniform distribution on $\mathfrak D_\epsilon$.  
Fix $1\le j\le L$ and
$\mathbf y,\mathbf z\in\mathcal S_j$.  
Let $\tau_{\rm out}$ be the
transposition of $\mathbf y$ and $\mathbf z$ in $\mathbb F_2^m$, and let $\tau_{\rm in}$ be the transposition of $B^\top\mathbf y$ and $B^\top\mathbf z$ in $\mathbb F_2^{2n}$.  
The map
\begin{equation*}
\operatorname{Dec}
\longmapsto
\tau_{\rm out}\circ\operatorname{Dec}\circ\tau_{\rm in}
\end{equation*}
is a bijection of $\mathfrak D_\epsilon$ onto itself.  
It exchanges the events that $\mathbf y$ and $\mathbf z$ are decoded correctly and leaves the
correctness status of every other error unchanged.  Consequently, under the uniform distribution on $\mathfrak D_\epsilon$, the probability of correct
decoding is constant within each Hamming-weight shell.  
Let $p_j$ denote this common probability on $\mathcal S_j$.  
The expected number of correctly decoded vectors in $\mathcal S_j$ is then $|\mathcal S_j|p_j$.  
By Eq. \eqref{260722eq1}, 
\begin{equation}\label{260722eq2}
p_j\ge 1-\epsilon,\qquad \forall 1\le j\le L.    
\end{equation} 

Write
\begin{equation*}
\widetilde w_{\mathbf y}
:=
\sum_{s=0}^l a_s\beta_G^{(s)}(\mathbf y,\mathbf c),
\qquad
\mathcal N^2
:=
\sum_{|\mathbf y|\le l}|\widetilde w_{\mathbf y}|^2,
\end{equation*}
where $G$ is the anticommutation graph of $H$, and set
$$
r_{\operatorname{Dec}}(\mathbf y)
:=
\mathbf y\oplus\operatorname{Dec}(B^\top\mathbf y).
$$
The state evolution in the proof of Theorem \ref{thm:hdqi-noncomm} gives the following purifications of the perfect and imperfect outputs:
\begin{align*}
\ket{\psi_{\CPP}}
&=
\frac{1}{\mathcal N}
\sum_{|\mathbf y|\le l}
\widetilde w_{\mathbf y}
\ket{\mathbf 0}_A
\otimes(P_{\mathbf y}\otimes I)\ket{\Phi_n}_{BC},
\\
\ket{\psi_{\CPP,\operatorname{Dec}}}
&=
\frac{1}{\mathcal N}
\sum_{|\mathbf y|\le l}
\widetilde w_{\mathbf y}
\ket{r_{\operatorname{Dec}}(\mathbf y)}_A
\otimes(P_{\mathbf y}\otimes I)\ket{\Phi_n}_{BC}.
\end{align*}
For $|\mathbf x|,|\mathbf y|\le l$, injectivity of the syndrome map gives
\begin{equation*}
\langle\Phi_n|
(P_{\mathbf x}^\dagger P_{\mathbf y}\otimes I)
|\Phi_n\rangle
=
\delta_{\mathbf x,\mathbf y}.
\end{equation*}
Set
\begin{equation*}
F_{\operatorname{Dec}}
:=
F\bigl(
\rho_{\CPP,\operatorname{Dec}}(H),
\rho_{\CPP}(H)
\bigr).
\end{equation*}
Uhlmann's theorem therefore implies
\begin{align*}
F_{\operatorname{Dec}}
&\ge
\frac{1}{\mathcal N^2}
\sum_{|\mathbf y|\le l}
|\widetilde w_{\mathbf y}|^2
\mathbf 1\!\left\{
r_{\operatorname{Dec}}(\mathbf y)=\mathbf 0
\right\}. 
\end{align*}
Averaging over $\mathfrak D_\epsilon$, 
and using Eq. \eqref{260722eq2} together with $\operatorname{Dec}(\mathbf 0)=\mathbf 0$, 
yields
\begin{align*}
\frac{1}{|\mathfrak D_\epsilon|}
\sum_{\operatorname{Dec}\in\mathfrak D_\epsilon}
F_{\operatorname{Dec}}
\ge
\frac{1-\epsilon}{\mathcal N^2}
\sum_{|\mathbf y|\le l}|\widetilde w_{\mathbf y}|^2
=1-\epsilon.
\end{align*}
Equation~\eqref{260722eq3} follows from the Fuchs--van de Graaf inequality and Jensen's inequality.
\end{proof}

\section{Conclusion and Outlook}
\label{sec:conclusion}
In this work, we investigate HDQI for general Pauli Hamiltonians of the form $H=\sum_i c_i P_i$,  where 
$c_i\in \real$, and $P_i$ are $n$-qubit Pauli operators, encompassing both commuting and noncommuting cases.
We give algorithms for preparing $\rho_{\CPP}(H)$, which can approximate the Gibbs state of $H$.
The algorithms are efficient when the required decoder and reference state can be implemented efficiently.
Moreover, we demonstrate that these algorithms are robust to imperfections in the decoding procedure.
Our results substantially extend the scope of HDQI beyond stabilizer-like  Hamiltonians and provide a unified framework applicable to a broad class of physically and computationally relevant models.

Several directions warrant further investigation. 
The first concerns semicircle laws and ground state preparation. 
For commuting Hamiltonians with $v_i \in \{\pm 1\}$, the expected energy under $\rho_{\CPP}(H)$ obeys a semicircle law~\cite{schmidhuber2025hamiltonian,jordan2024optimization}.
Whether this extends to general real coefficients and to noncommuting Hamiltonians is open, the main obstacles being the absence of mathematical tools for noncommuting Pauli expansions and the lack of an optimal choice of polynomial $\CPP(x)$. 
Since the Gibbs state converges to the ground-state projector in the
zero-temperature limit \(\beta\to\infty\), HDQI may offer a new route to ground-state preparation for many-body Hamiltonians.

A second direction is the extension of HDQI to continuous-variable (CV) 
systems, where qubits are replaced by bosonic modes and Pauli operators by 
displacement or quadrature operators. The discrete symplectic structure 
$\mathbb{F}_2^{2n}$ underlying our framework has a natural analog in the 
continuous symplectic group $\mathrm{Sp}(2n, \mathbb{R})$ acting on bosonic 
phase space. Such an extension would open HDQI to applications in 
quantum optics, photonic processors, and bosonic many-body physics.
We leave this direction for future work.

\section*{Acknowledgments}
K. Bu thanks Anurag Anshu, Arthur Jaffe, Stephen Jordan,  Seth Lloyd, Peter Love, Feng Qian, Alexander Schmidhuber  for helpful discussions.
K. Bu is partly supported by the NSF
grant DMS-2606749, JobsOhio GR138220, and ARO Grant W911NF19-1-0302 and the ARO MURI Grant W911NF-20-1-0082.
%


\bibliographystyle{unsrt}
\bibliography{reference}{}

\end{document}